\documentclass[12pt,fleqn]{article}
\usepackage{amsmath,amssymb,amsthm,tikz,indentfirst,graphicx,mathrsfs,enumerate,epstopdf,caption,subcaption,float,authblk}\usepackage{enumitem}
\usepackage[square,comma,numbers,sort&compress]{natbib}
\usepackage[colorlinks]{hyperref}
\usepackage{multirow}
\theoremstyle{plain}
\newtheorem{theorem}{Theorem}[section]
\newtheorem{corollary}[theorem]{Corollary}
\newtheorem{lemma}[theorem]{Lemma}
\newtheorem{assume}[theorem]{Assumption}
\newtheorem{prop}[theorem]{Proposition}
\theoremstyle{remark}
\newtheorem{remark}[theorem]{Remark}
\theoremstyle{definition}

\newtheorem{problem}[theorem]{Riemann--Hilbert Problem}
\usepackage[left=2.5cm,right=2.5cm,bottom=2.8cm,top=2.8cm]{geometry}

\newcommand{\res}{\operatorname{Res}\ }

\numberwithin{equation}{section}
\allowdisplaybreaks[2]
\begin{document}
	\title{ Inverse Scattering Transform  for the Massive Thirring Model: Delving into  Higher-Order Pole Dynamics}
	
	\author{Dongli Luan}
	\author{Bo Xue}
	\author{Huan Liu \thanks{Corresponding author. E-mail: liuhuan@zzu.edu.cn}}
	\affil{School of Mathematics and Statistics, Zhengzhou University, Zhengzhou, Henan 450001, People's  Republic  of China}
	\renewcommand*{\Affilfont}{\small\it}
	\renewcommand{\Authands}{, }
	\date{}
	
	\maketitle
	\begin{abstract}
		
We investigate the inverse scattering problem for the massive Thirring model, focusing particularly on cases where the transmission coefficient exhibits $N$ pairs of higher-order poles. Our methodology involves transforming initial data into scattering data via the direct scattering problem. Utilizing two parameter transformations, we examine the asymptotic properties of the Jost functions at both vanishing and infinite parameters, yielding two equivalent spectral problems. We subsequently devise a mapping that translates the obtained scattering data into a $2 \times 2$ matrix Riemann--Hilbert problem, incorporating several residue conditions at $N$ pairs of multiple poles. Additionally, we construct an equivalent pole-free Riemann--Hilbert problem and demonstrate the existence and uniqueness of its solution. In the reflectionless case, the $N$-multipole solutions can be reconstructed by resolving two linear algebraic systems.
	\end{abstract}
	
	\newpage
	\section{Introduction}
	The massive Thirring model (MTM), described by the following system of equations:
	\begin{equation}\label{mtm}
		\left\{\begin{array}{c}\mathrm{i}\left(u_t+u_x\right)+v+|v|^2u=0,\\\mathrm{i}\left(v_t-v_x\right)+u+|u|^2v=0,\end{array}\right.\end{equation}
	is a field theory extension of the Thirring model first proposed by Walter Thirring in 1958 \cite{Thirring} to describe interacting fermions with mass. It was one of the first examples of a 1+1-dimensional quantum field theory that exhibited non-trivial interacting fermions with mass contrasting the massless Dirac fermion model \cite{massless1,massless2}.  The integrability of the MTM \eqref{mtm} has been established through the work of various researchers, including Kuznetsov and Mikhailov \cite{integrability1}, Mikhailov \cite{integrability2}, Orfanidis \cite{integrability3}, and Kaup and Newell \cite{integrability4}. The MTM \eqref{mtm} is intrinsically connected to sine-Gordon theory, and there is a substantial body of research exploring the correspondences between these theories \cite{sine-mtm2,sine-mtm5,sine-mtm6}.  Extensive investigations have also been conducted on the conservation laws \cite{conservation1}, asymptotic behavior \cite{asymptotic}, and various structural aspects \cite{structure,structure2,str1,str2,str3,str4}.	
	As a derivative-type nonlinear Schr\"odinger equation \cite{derivative,der2}, the MTM \eqref{mtm} has attracted numerous researchers interested in finding its soliton solutions \cite{integrability1,integrability3,integrability4,solution1,solution2,solution3,solution}. The Darboux and B\"acklund transformations of the MTM \eqref{mtm} have been explored in \cite{backlund1,backlund2,Darboux1,Darboux2}. The direct method for the MTM \eqref{mtm} has been detailed in \cite{direct}, while the inverse scattering transform for the MTM \eqref{mtm} has been examined by Pelinovsky and Saalmann \cite{inverse} and others \cite{inverse4,integrability1,inverse1,inverse2,liu}. However, the extant scholarly works have predominantly concentrated their efforts on the pursuit of simple pole solutions to the MTM \eqref{mtm}, with the investigation into the complexities of multiple pole solutions remaining unexplored.

In this paper, we are preparing to study the inverse scattering transform of the MTM \eqref{mtm} in the case of multiple poles.  Although  the Riemann--Hilbert (RH) method has been successfully applied to construct $N$-multipole solutions for integrable equations \cite{sasa,AL,SBE}, the spectral structure associated with the MTM equation \eqref{mtm} does not exhibit the same degree of symmetry as those encountered in the equations previously mentioned.	To address the symmetry issue, we introduce two parameter transformations and employ one of them to construct the RH problem (refer to RH problem \ref{pro}), which is furnished with residue conditions at certain higher-order poles. The existence and uniqueness  of the solution to this RH problem are rigorously demonstrated (see Vanishing Lemma \ref{vanish}). Furthermore, the reconstruction formula for the solutions to the MTM equation \eqref{mtm} from the existing literature is generally presented as a quotient of two expressions, necessitating two distinct parameter transformations and consequently giving rise to two separate RH problems. This dual approach renders the process of deriving exact solutions somewhat more intricate and laborious. We streamline the procedure by obtaining the solution via a dual single-limit derivation (refer to Corollary \ref{recon}) and by resolving a single RH problem equipped with data derived from a single parameter transformation.

	The structure of the remainder of this paper is as follows. In Section $\ref{sec:dir}$, we introduce a mapping from the initial data to the scattering data, and  provide an analysis of the discrete spectrum corresponding to multiple zeros. We introduce two parameter transformations to obtain two modified Jost solutions and analyze their pertinent characteristics. In Section $\ref{sec:inverse}$, we construct another mapping from the scattering data to a $2\times2$ matrix RH problem equipped with several residue conditions set at $N$ pairs of multiple poles. Thereafter,  we construct an equivalent pole-free RH problem and prove the unique   solvability  of  this problem. Ultimately, in the reflectionless case, $N$-multipole solutions can be reconstructed through resolving two linear algebraic systems.
	
	Throughout this paper, we adopt the following standard notations: The complex conjugate of a complex number $\lambda$ is denoted by $\bar{\lambda}$ and for a complex-valued matrix $\mathbf{X}$, $\bar{\mathbf{X}}$ represents its
	element-wise complex conjugate. The transpose of $\mathbf{X}$ is denoted by $\mathbf{X}^{\mathrm{T}}$, with $\mathbf{X}_j$ indicating the $j$-th column and $\mathbf{X}_{ij}$ refers to the $(i,j)$-entry. Given $2 \times 2$ matrices $\mathbf{A}$ and $\mathbf{B}$, their commutator is defined as $[\mathbf{A},\mathbf{B}]=\mathbf{AB}-\mathbf{BA}$. The $2 \times 2$ identity matrix is denoted by $\mathbf{I}$. For a matrix-valued function $\mathbf{Y}(x)=(a_{ij}(x))_{2\times 2}$ defined on the contour $\Gamma$, the $L^1$-norm and $L^{\infty}$-norm of $\mathbf{Y}(x)$ are defined as:
	\begin{equation}
\begin{aligned}
&\|\mathbf{Y}(x)\|_{L^1(\Gamma)}=\int_{\Gamma}\|\mathbf{Y}\|_1(x)\mathrm{d}x,\\ &\|\mathbf{Y}(x)\|_{L^{\infty}(\Gamma)}=\sup\{a\in\mathbb{R}:m(\{x\in\Gamma:\|\mathbf{Y}\|_\infty(x)>a\}) >0\},
\end{aligned}
 \end{equation}
where  $\|\cdot\|_1$ and $\|\cdot\|_\infty$  refers to $1$-norm and $\infty$-norm of the matrix, respectively, and $m(\cdot)$ represents the Lebesgue measure.
  Additionally, we employ the following three types of Pauli matrices:
	\begin{equation}
		\sigma_1=\left(\begin{array}{cc}0&1\\1&0\end{array}\right),\quad 	\sigma_2=\left(\begin{array}{cc}0&-\mathrm{i}\\\mathrm{i}&0\end{array}\right),\quad	\sigma_3=\left(\begin{array}{cc}1&0\\0&-1\end{array}\right).
	\end{equation}
	The notation $\mathbb{C}^{\pm\pm}$ is used to represent the four quadrants of the complex plane, where the first superscript indicates the sign of the real part and the second superscript indicates the sign of the imaginary part. To enhance understanding, a schematic diagram is provided in Figure \ref{Cppmm}.
	\begin{figure}[!htbp]
		\centering	
		\begin{tikzpicture}
			% Draw the axes
			\draw[->] (-2,0) -- (2,0) node[right] {Re};
			\draw[->] (0,-2) -- (0,2) node[above] {Im};
			% Label the quadrants
			\node at (1,1) {$\mathbb{C}^{++}$};
			\node at (-1,1) {$\mathbb{C}^{-+}$};
			\node at (-1,-1) {$\mathbb{C}^{--}$};
			\node at (1,-1) {$\mathbb{C}^{+-}$};
		\end{tikzpicture}
\caption{The domain $\mathbb{C}^{\pm\pm}$.}\label{Cppmm}
	\end{figure}
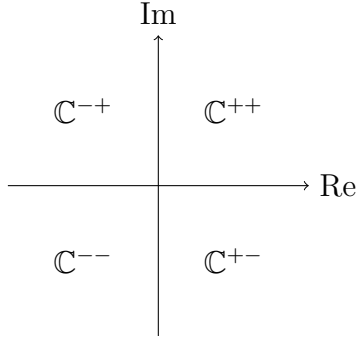

	\section{Direct scattering problem }\label{sec:dir}
	The MTM \eqref{mtm} admits the Lax pair:
	\begin{equation}\label{lax}
		\phi_{x}=\mathbf{U}\phi,\quad\phi_{t}=\mathbf{V}\phi,
	\end{equation}
	with
	\begin{equation}
		\begin{aligned}
			&\mathbf{U}=\frac{\mathrm{i}}{4}\left(|u|^{2}-|v|^{2}\right)\sigma_{3}-\frac{\mathrm{i}\zeta}{2}\left(\begin{array}{cc}0&\bar{v}\\v&0\end{array}\right)+\frac{\mathrm{i}}{2}\zeta^{-1}\left(\begin{array}{cc}0&\bar{u}\\u&0\end{array}\right)+\frac{\mathrm{i}}{4}\left(\zeta^{2}-\zeta^{-2}\right)\sigma_{3},\\&\mathbf{V}=-\frac{\mathrm{i}}{4}\left(|u|^{2}+|v|^{2}\right)\sigma_{3}-\frac{\mathrm{i}\zeta}{2}\left(\begin{array}{cc}0&\bar{v}\\v&0\end{array}\right)-\frac{\mathrm{i}}{2}\zeta^{-1}\left(\begin{array}{cc}0&\bar{u}\\u&0\end{array}\right)+\frac{\mathrm{i}}{4}\left(\zeta^{2}+\zeta^{-2}\right)\sigma_{3},\end{aligned}\end{equation}
	where the function $\phi=\phi(x,t;\zeta)$ is a $2\times2$ matrix-valued function of $x$, $t$, and the spectral parameter $\zeta\in\mathbb{C}$. It should be noted that the MTM \eqref{mtm} is equivalent to  the following compatibility condition:
	\begin{equation}
		\mathbf{U}_t-\mathbf{V}_x+\left[\mathbf{U},\mathbf{V}\right]=\mathbf{0}.
	\end{equation}

	\subsection{Jost solutions}
	We aim to find functions $u(x,t)$ and $v(x,t)$ that are solutions to Eq.\,\eqref{mtm} for $(x,t)\in\mathbb{R}\times\mathbb{R}^{+}$. These functions are required to decay to zero as $x$ approaches infinity, with their $x$-derivatives also vanishing at infinity. Additionally, they must adhere to the specific initial conditions, and $u(x,0), v(x,0)$ are given functions belonging to an appropriate function space. Our objective is to determine two Jost solutions, $\phi_{\pm}(x,t;\zeta)$, which are to fulfill the following boundary conditions:
	\begin{equation}\label{psi}
		\phi_{\pm}(x,t;\zeta)=\mathrm{e}^{\mathrm{i}\theta(x,t;\zeta)\sigma_3}+o(1),\quad x\to\pm\infty,
	\end{equation}
	where $\theta(x,t;\zeta)=\frac{1}{4}\left(\zeta^2-\zeta^{-2}\right)x+\frac{1}{4}\left(\zeta^2+\zeta^{-2}\right)t$. More precisely,
	\begin{equation}\phi_{-1}(x,t;\zeta)\sim\begin{pmatrix}1\\0\end{pmatrix}\mathrm{e}^{\mathrm{i}\theta(x,t;\zeta)},\quad\phi_{-2}(x,t;\zeta)\sim\begin{pmatrix}0\\1\end{pmatrix}\mathrm{e}^{-\mathrm{i}\theta(x,t;\zeta)},\quad x\to-\infty,
	\end{equation}	and
	\begin{equation}	
		\phi_{+1}(x,t;\zeta)\sim\begin{pmatrix}1\\0\end{pmatrix}\mathrm{e}^{\mathrm{i}\theta(x,t;\zeta)},\quad\phi_{+2}(x,t;\zeta)\sim\begin{pmatrix}0\\1\end{pmatrix}\mathrm{e}^{-\mathrm{i}\theta(x,t;\zeta)},\quad x\to+\infty. \end{equation}
	Define the normalized Jost solutions as
	\begin{equation}\label{m}
		\varphi_{\pm}(x,t;\zeta)=\phi_{\pm}(x,t;\zeta)\mathrm{e}^{-\mathrm{i}\theta(x,t;\zeta)\sigma_3},
	\end{equation}
	which satisfy the differential equation and constant boundary conditions at infinity:
	\begin{equation}\label{partial}
		\partial_{x}\varphi_{\pm}=\frac{\mathrm{i}}{4}(\zeta^2-\zeta^{-2})[\sigma_{3},\varphi_{\pm}]+\mathbf{Q}\varphi_{\pm},\quad\lim_{x\to\pm\infty}\varphi_{\pm}=\mathbf{I},
	\end{equation}
	where the dependence on the variables $ (x, t; \zeta) $ is omitted here for conciseness and the matrix $\mathbf{Q} $ is defined as
	\begin{equation}
		\mathbf{Q}=\mathbf{Q}(x,t;\zeta)=\frac{\mathrm{i}}{4}\left(|u|^{2}-|v|^{2}\right)\sigma_{3}-\frac{\mathrm{i}\zeta}{2}\left(\begin{array}{cc}0&\bar{v}\\v&0\end{array}\right)+\frac{\mathrm{i}}{2}\zeta^{-1}\left(\begin{array}{cc}0&\bar{u}\\u&0\end{array}\right).
	\end{equation}
	The integration of Eq.\,\eqref{partial} results in the following  Volterra integral equations:
	\begin{equation}\label{m:volterra}
		\begin{aligned}
			\varphi_{\pm}(x,t;\zeta)=&\left(\varphi_{\pm1}(x,t;\zeta),\varphi_{\pm2}(x,t;\zeta)\right)\\
			=&\mathbf{I}+\int_{\pm\infty}^{x}\mathrm{e}^{\frac{\mathrm{i}}{4}(\zeta^{2}-\zeta^{-2})(x-y)\hat{\sigma}_{3}}\left(\mathbf{Q}\varphi_{\pm}\right)(y,t;\zeta)\mathrm{d}y,
		\end{aligned}
	\end{equation}where $\hat{\sigma}_{3}\mathbf{X}=[\sigma_{3},\mathbf{X}]$, consequently, $\mathrm{e}^{\hat{\sigma}_{3}}\mathbf{X}=\mathrm{e}^{\sigma_{3}}\mathbf{X}\mathrm{e}^{-\sigma_{3}}.$
	
	A standard assumption in analyzing Volterra integral equations \eqref{m:volterra} is that $u(\cdot,t), v(\cdot,t)$\\
	$\in L^1(\mathbb{R})\cap L^2(\mathbb{R})$ for fixed $t$. However, even under these conditions, $\|\mathbf{Q}(\cdot,t;\zeta)\|_{L^1(\mathbb{R})}$ does not uniformly control as $\zeta\to0$ and $\zeta\to\infty$. This lack of uniform convergence presents challenges in studying the behavior of $\varphi_{\pm 1}(x,t;\zeta)$ and $\varphi_{\pm 2}(x,t;\zeta)$ as $\zeta\to 0$ and $\zeta\to\infty$. To address this issue, we employ two transformations, facilitating the analysis of the spectral problem under varying asymptotic conditions of $\zeta$.
	\subsubsection{Transformation of the Jost solutions for large $\zeta$}
		Define $\lambda=\zeta^2$ and the domains
	\begin{equation}
		\Omega_0=\{\lambda\in\mathbb{C}:0<|\lambda|<1\},\quad \Omega_1=\{\lambda\in\mathbb{C}:|\lambda|=1\},\quad \Omega_{\infty}=\{\lambda\in\mathbb{C}:|\lambda|>1\}.
	\end{equation}
	Assume $v(\cdot,t)\in L^{\infty}(\mathbb{R})$ for fixed $t$, let the transformation matrix $\mathbf{T}(x,t;\lambda)$ be
	\begin{equation}\mathbf{T}(x,t;\lambda)=\left(\begin{array}{cc}1&0\\v&\lambda^{\frac12}\end{array}\right).\end{equation}
	The normalized Jost solutions are defined as follows:
	\begin{equation}\label{m3}\begin{aligned}&\Phi_{\pm1}(x,t;\lambda)=\mathbf{T}(x,t;\lambda)\varphi_{\pm1}(x,t;\zeta),\\&\Phi_{\pm2}(x,t;\lambda)=\zeta^{-1}
			\mathbf{T}(x,t;\lambda)\varphi_{\pm2}(x,t;\zeta),\end{aligned}\end{equation}
	which satisfy the differential equation and constant boundary conditions at infinity:
	\begin{equation}\label{par3}
		\partial_{x}\Phi_{\pm}=\frac{\mathrm{i}}{4}(\lambda-\lambda^{-1})[\sigma_{3},\Phi_{\pm}]+\left(\mathbf{Q}_1+\lambda^{-1} \mathbf{Q}_2\right)\Phi_{\pm},\quad\lim_{x\to\pm\infty}\Phi_{\pm}=\mathbf{I},
	\end{equation}
	where
	\begin{equation}\left.\begin{aligned}&\mathbf{Q}_{1}(x,t)=\left(\begin{array}{cc}\frac{\mathrm{i}}{4}\left(|u|^{2}+|v|^{2}\right)&-\frac{\mathrm{i}}{2}\bar{v}\\v_{x}+\frac{\mathrm{i}}{2}|u|^{2}v+\frac{\mathrm{i}}{2}u&-\frac{\mathrm{i}}{4}\left(|u|^{2}+|v|^{2}\right)\\\end{array}\right),\\&\mathbf{Q}_{2}(x,t)=-\frac{\mathrm{i}}{2}\left(\begin{array}{cc}\bar{u}v&-\bar{u}\\v+\bar{u}v^{2}&-\bar{u}v\end{array}\right).\end{aligned}\right.\end{equation}	
	Eq.\,\eqref{par3} can be integrated to yield the following Volterra integral equations:
	\begin{equation}\begin{aligned}\label{vo3}
			\Phi_{\pm}(x,t;\lambda)& =\left(\Phi_{\pm1}(x,t;\lambda),\Phi_{\pm2}(x,t;\lambda)\right) \\
			&=\mathbf{I}+\int_{\pm\infty}^x\mathrm{e}^{\frac{\mathrm{i}}{4}(\lambda-\lambda^{-1})(x-y)\hat{\sigma}_3}\left[(\mathbf{Q}_1(y,t)+\lambda^{-1} \mathbf{Q}_2(y,t))\Phi_{\pm}(y,t;\lambda)\right]\mathrm{d}y.
	\end{aligned}\end{equation}
	
	According to standard Volterra theory, we present the following theorem:
	\begin{theorem}\label{th4}
		Under the condition that $u(\cdot,t), v(\cdot,t)\in L^1(\mathbb{R})\cap L^{\infty}(\mathbb{R})$ and $v_x(\cdot,t)\in L^1(\mathbb{R})$ for fixed $t$, the normalized Jost solutions $\Phi_{\pm}(x,t;\lambda)$ in Eq.\,\eqref{vo3} are well-defined for $\lambda\in \mathbb{R}\backslash[-1,1]$. Specifically, for every $x\in\mathbb{R}$, $\Phi_{+1}(x,t;\cdot)$ and $\Phi_{-2}(x,t;\cdot)$ can be analytically continued onto the regions $\mathbb{C}^{+}\cap\Omega_{\infty}$, while $\Phi_{-1}(x,t;\cdot)$ and $\Phi_{+2}(x,t;\cdot)$ can be analytically continued onto $\mathbb{C}^{-}\cap\Omega_{\infty}$. Furthermore, the Volterra intergal equations \eqref{vo3} admit unique solutions $\Phi_{\pm1}(\cdot,t;\lambda)$ and $\Phi_{\pm2}(\cdot,t;\lambda)$ in $L^{\infty}(\mathbb{R})$ for every $\lambda\in \mathbb{R}\backslash[-1,1]$. The functions $\Phi_{\pm1}(x,t;\lambda)$ and $\Phi_{\pm2}(x,t;\lambda)$ exhibit the following asymptotic behavior as $\lambda\to\infty$ in the domains of their analyticity:
		\begin{equation}\label{lim2}
			\lim_{\lambda\to\infty}\frac{\Phi_{\pm1}(x,t;\lambda)}{\Phi_{\pm1}^{\infty}(x,t)}=\mathrm{\mathbf{e}}_1=(1,0)^{\mathrm{T}},\quad
			\lim_{\lambda\to\infty}\frac{\Phi_{\pm2}(x,t;\lambda)}{\Phi_{\pm2}^{\infty}(x,t)}=\mathrm{\mathbf{e}}_2=(0,1)^{\mathrm{T}},
		\end{equation}
		where
		\begin{equation}\Phi_{\pm1}^{\infty}(x,t)=\mathrm{e}^{\frac {\mathrm{i}}{4}\int_{\pm\infty}^x\left(|u|^2+|v|^2\right)\mathrm{d}y},\quad \Phi_{\pm2}^{\infty}(x,t)=\mathrm{e}^{-\frac {\mathrm{i}}{4}\int_{\pm\infty}^x\left(|u|^2+|v|^2\right)\mathrm{d}y}.\end{equation}
		If $v(\cdot,t)\in C^1(\mathbb{R})$, then
		\begin{equation}\left.\begin{aligned}&\lim_{\lambda\to\infty}\lambda\left[\frac{\Phi_{\pm1}(x,t;\lambda)}{\Phi_{\pm1}^{\infty}(x,t)}-\mathrm{\mathbf{e}}_1\right]=\left(\begin{array}{c}-\int_{\pm\infty}^{x}\left[\bar{v}\left(v_x+\frac{\mathrm{i}}{2}|u|^2v+\frac{\mathrm{i}}{2}u\right)+\frac{\mathrm{i}}{2}\bar{u}v\right]\mathrm{d}y\\-2\mathrm{i}v_x+|u|^2v+u\end{array}\right),\\&\lim_{\lambda\to\infty}\lambda\left[\frac{\Phi_{\pm2}(x,t;\lambda)}{\Phi_{\pm2}^{\infty}(x,t)}-\mathrm{\mathbf{e}}_2\right]=\left(\begin{array}{c}\bar{v}\\\int_{\pm\infty}^{x}\left[\bar{v}\left(v_x+\frac{\mathrm{i}}{2}|u|^2v+\frac{\mathrm{i}}{2}u\right)+\frac{\mathrm{i}}{2}\bar{u}v\right]\mathrm{d}y\end{array}\right.\bigg).\end{aligned}\right.\end{equation}
	\end{theorem}
	\subsubsection{Transformation of the Jost solutions for small $\zeta$}
	Assume $u(\cdot,t)\in L^{\infty}(\mathbb{R})$ for fixed $t$ and $\zeta\neq0$, let the transformation matrix $\widetilde{\mathbf{T}}(x,t;\lambda)$ be
	\begin{equation}\widetilde{\mathbf{T}}(x,t;\lambda)=\left(\begin{array}{cc}1&0\\u&\lambda^{-\frac12}\end{array}\right).\end{equation}
	Consistent with the previous discussion regarding large $\zeta$, we define the new Jost solutions as follows:
	\begin{equation}
		\begin{aligned}
			&\widetilde{\Phi}_{\pm1}(x,t;\lambda)=\widetilde{\mathbf{T}}(x,t;\lambda)\varphi_{\pm1}(x,t;\zeta),\\
			&\widetilde{\Phi}_{\pm2}(x,t;\lambda)=\zeta \widetilde{\mathbf{T}}(x,t;\lambda)\varphi_{\pm2}(x,t;\zeta),
		\end{aligned}
	\end{equation}
	which satisfy the differential equation and constant boundary conditions at infinity:
	\begin{equation}\label{par2}
		\partial_{x}\widetilde{\Phi}_{\pm}=\frac{\mathrm{i}}{4}(\lambda-\lambda^{-1})[\sigma_{3},\widetilde{\Phi}_{\pm}]+\left(\widetilde{\mathbf{Q}}_1+\lambda \widetilde{\mathbf{Q}}_2\right)\widetilde{\Phi}_{\pm},\quad \lim_{x\to\pm\infty}\widetilde{\Phi}_{\pm}=\mathbf{I},
	\end{equation}
	where
	\begin{equation}
		\begin{aligned}
			&\widetilde{\mathbf{Q}}_1(x,t)=\left(\begin{array}{cc}-\frac{\mathrm{i}}{4}\left(|u|^2+|v|^2\right)&\frac{\mathrm{i}}{2}\bar{u}\\u_x-\frac{\mathrm{i}}{2}|v|^2u-\frac{\mathrm{i}}{2}v&\frac{\mathrm{i}}{4}\left(|u|^2+|v|^2\right)\end{array}\right),\\
			&\widetilde{\mathbf{Q}}_2(x,t)=\frac{\mathrm{i}}{2}\left(\begin{array}{cc}u\bar{v}&-\bar{v}\\u+u^2\bar{v}&-u\bar{v}\end{array}\right).
		\end{aligned}
	\end{equation}
	Eq.\,\eqref{par2} is integrated to derive the following Volterra integral equations:
	\begin{equation}\begin{aligned}\label{vo2}
			\widetilde{\Phi}_{\pm}(x,t;\lambda)& =\left(\widetilde{\Phi}_{\pm1}(x,t;\lambda),\widetilde{\Phi}_{\pm2}(x,t;\lambda)\right) \\
			&=\mathbf{I}+\int_{\pm\infty}^x\mathrm{e}^{\frac{\mathrm{i}}{4}(\lambda-\lambda^{-1})(x-y)\hat{\sigma}_3}\left[(\widetilde{\mathbf{Q}}_1(y,t)+\lambda \widetilde{\mathbf{Q}}_2(y,t))\widetilde{\Phi}_{\pm}(y,t;\lambda)\right]\mathrm{d}y.
	\end{aligned}\end{equation}
	
	According to standard Volterra theory, we state the following theorem:
	\begin{theorem}\label{3}
		Under the condition that $u(\cdot,t), v(\cdot,t)\in L^1(\mathbb{R})\cap L^{\infty}(\mathbb{R})$ and $u_x(\cdot,t)\in L^1(\mathbb{R})$ for fixed $t$, the normalized Jost solutions $\widetilde{\Phi}_{\pm}(x,t;\lambda)$ in Eq.\,\eqref{vo2} are well-defined for $\lambda\in (-1,0)\cup(0,1)$. Specifically, for every $x\in\mathbb{R}$, $\widetilde{\Phi}_{+1}(x,t;\cdot)$ and $\widetilde{\Phi}_{-2}(x,t;\cdot)$ can be analytically continued onto the regions $\mathbb{C}^{+}\cap\Omega_{0}$, while $\widetilde{\Phi}_{-1}(x,t;\cdot)$ and $\widetilde{\Phi}_{+2}(x,t;\cdot)$ can be analytically continued onto $\mathbb{C}^{-}\cap\Omega_{0}$. Furthermore, the Volterra intergal equations \eqref{vo2} admit unique solutions $\widetilde{\Phi}_{\pm1}(\cdot,t;\lambda)$ and $\widetilde{\Phi}_{\pm2}(\cdot,t;\lambda)$ in $L^{\infty}(\mathbb{R})$ for every $\lambda\in (-1,0)\cup(0,1)$.
		 The functions $\widetilde{\Phi}_{\pm1}(x,t;\lambda)$ and $\widetilde{\Phi}_{\pm2}(x,t;\lambda)$ exhibit the following asymptotic behavior as $\lambda\to0$ in the domains of their analyticity:
		\begin{equation}\label{lim1}
			\begin{aligned}
				\lim_{\lambda\to0}\frac{\widetilde{\Phi}_{\pm1}(x,t;\lambda)}{\widetilde{\Phi}_{\pm1}^{0}(x,t)}=\mathrm{\mathbf{e}}_1=(1,0)^{\mathrm{T}},\quad\lim_{\lambda\to0}\frac{\widetilde{\Phi}_{\pm2}(x,t;\lambda)}{\widetilde{\Phi}_{\pm2}^{0}(x,t)}=\mathrm{\mathbf{e}}_2=(0,1)^{\mathrm{T}},
		\end{aligned}\end{equation}
		where
		\begin{equation}\begin{aligned}
				\widetilde{\Phi}_{\pm1}^{0}(x,t)=\mathrm{e}^{-\frac {\mathrm{i}}{4}\int_{\pm\infty}^x\left(|u|^2+|v|^2\right)\mathrm{d}y},\quad \widetilde{\Phi}_{\pm2}^{0}(x,t)=\mathrm{e}^{\frac {\mathrm{i}}{4}\int_{\pm\infty}^x\left(|u|^2+|v|^2\right)\mathrm{d}y}.
		\end{aligned}\end{equation}
		If $u(\cdot,t)\in C^1(\mathbb{R})$, then
		\begin{equation}
			\begin{aligned}
				&\lim_{\lambda\to0}\lambda^{-1}\left[\frac{\widetilde{\Phi}_{\pm1}(x,t;\lambda)}{\widetilde{\Phi}_{\pm1}^{0}(x,t)}-\mathrm{\mathbf{e}}_1\right]=\left(\begin{array}{c}-\int_{\pm\infty}^x\left[\bar{u}\left(u_x-\frac{\mathrm{i}}{2}u|v|^2-\frac{\mathrm{i}}{2}v\right)-\frac{\mathrm{i}}{2}u\bar{v}\right]\mathrm{d}y\\2\mathrm{i}u_x+u|v|^2+v\end{array}\right),\\&\lim_{\lambda\to0}\lambda^{-1}\left[\frac{\widetilde{\Phi}_{\pm2}(x,t;\lambda)}{\widetilde{\Phi}_{\pm2}^{0}(x,t)}-\mathrm{\mathbf{e}}_2\right]=\left(\begin{array}{c}\bar{u}\\\int_{\pm\infty}^x\left[\bar{u}\left(u_x-\frac{\mathrm{i}}{2}u|v|^2-\frac{\mathrm{i}}{2}v\right)-\frac{\mathrm{i}}{2}u\bar{v}\right]\mathrm{d}y\end{array}\right).
			\end{aligned}
	\end{equation}\end{theorem}
	\subsubsection{Connection between two pairs of normalized Jost solutions}
	In Theorems $\ref{th4}$ and $\ref{3}$, we establish the existence of two sets of Jost solutions,
	\begin{equation}\begin{aligned}&\{\Phi_{\pm1}(x,t;\lambda),\Phi_{\pm2}(x,t;\lambda)\},\quad \lambda\in\Omega_{\infty},\\ &\{\widetilde{\Phi}_{\pm1}(x,t;\lambda),\widetilde{\Phi}_{\pm2}(x,t;\lambda)\},\quad \lambda\in\Omega_0,
	\end{aligned}\end{equation}
	which are derived from the transformations of the Jost solutions
	$\left\{\phi_{\pm1}(x,t;\zeta),\phi_{\pm2}(x,t;\zeta)\right\}$ associated with the original Lax pair \eqref{lax}. These transformations establish the relationships between the sets of Jost solutions. For every $\lambda\in\Omega_1$, the following relationships hold:
	\begin{equation}\label{connect1}
		\begin{aligned}
			&\widetilde{\Phi}_{\pm1}(x,t;\lambda)=\left(\begin{array}{cc}1&0\\u(x,t)-\lambda^{-1}v(x,t)&\lambda^{-1}\end{array}\right)\Phi_{\pm1}(x,t;\lambda),\\\\&\widetilde{\Phi}_{\pm2}(x,t;\lambda)=\left(\begin{array}{cc}\lambda&0\\u(x,t)\lambda-v(x,t)&1\end{array}\right)\Phi_{\pm2}(x,t;\lambda),
	\end{aligned}\end{equation}
	conversely,
	\begin{equation}\label{connect2}\begin{aligned}\Phi_{\pm1}(x,t;\lambda)&=\left(\begin{array}{cc}1&0\\v(x,t)-\lambda u(x,t)&\lambda\end{array}\right)\widetilde{\Phi}_{\pm1}(x,t;\lambda),\\\Phi_{\pm2}(x,t;\lambda)&=\left(\begin{array}{cc}\lambda^{-1}&0\\v(x,t)\lambda^{-1}-u(x,t)&1\end{array}\right)\widetilde{\Phi}_{\pm2}(x,t;\lambda).\end{aligned}\end{equation}
	Based on Theorem $\ref{th4}$ and Eqs.\,\eqref{connect1}, we conclude that $\widetilde{\Phi}_{+1}(x,t;\cdot)$ and $\widetilde{\Phi}_{-2}(x,t;\cdot)$ can be analytically continued onto the regions $\mathbb{C}^{+}\cap\Omega_{\infty}$, whereas $\widetilde{\Phi}_{-1}(x,t;\cdot)$ and $\widetilde{\Phi}_{+2}(x,t;\cdot)$ can be analytically continued onto $\mathbb{C}^{-}\cap\Omega_{\infty}$. Similiarly, according to Theorem $\ref{3}$ and Eq.\,\eqref{connect2}, we can establish that $\Phi_{+1}(x,t;\cdot)$ and $\Phi_{-2}(x,t;\cdot)$ can be analytically continued onto the regions $\mathbb{C}^{+}\cap\Omega_{0}$, whereas $\Phi_{-1}(x,t;\cdot)$ and $\Phi_{+2}(x,t;\cdot)$ can be analytically continued onto $\mathbb{C}^{-}\cap\Omega_{0}$. Moreover, we get the existence of the following limits within their respective analytic domains from Theorems $\ref{th4}$ and $\ref{3}$:
	\begin{equation}\label{connect3}
		\begin{aligned}&\lim_{\lambda\to\infty}\frac{\widetilde{\Phi}_{\pm1}(x,t;\lambda)}{\Phi_{\pm1}^{\infty}(x,t)} =\mathrm{\mathbf{e}}_1+u(x,t)\mathrm{\mathbf{e}}_2, \\
			&\lim_{\lambda\to\infty}\frac{\widetilde{\Phi}_{\pm2}(x,t;\lambda)}{\Phi_{\pm2}^{\infty}(x,t)} =\overline{v(x,t)}\mathrm{\mathbf{e}}_1+\left(1+u(x,t)\overline{v(x,t)}\right)\mathrm{\mathbf{e}}_2,\\
			&\lim_{\lambda\to0}\frac{\Phi_{\pm1}(x,t;\lambda)}{\widetilde{\Phi}_{\pm1}^{0}(x,t)}=\mathrm{\mathbf{e}}_1+v(x,t)\mathrm{\mathbf{e}}_2,\\
			&\lim_{\lambda\to0}\frac{\Phi_{\pm2}(x,t;\lambda)}{\widetilde{\Phi}_{\pm2}^{0}(x,t)}=\overline{u(x,t)}\mathrm{\mathbf{e}}_1+\left(1+\overline{u(x,t)}v(x,t)\right)\mathrm{\mathbf{e}}_2.
	\end{aligned}\end{equation}
	
	By combining Theorem $\ref{th4}$ and  $\ref{3}$ with Eqs.\,\eqref{connect1} and \eqref{connect2}, we derive the following theorem:
	\begin{theorem}\label{th3}
		Under the condition that $u(\cdot,t), v(\cdot,t)\in L^1(\mathbb{R})\cap L^{\infty}(\mathbb{R})$ and $u_x(\cdot,t)$, $v_x(\cdot,t)\in L^1(\mathbb{R})$ for fixed $t$, the normalized Jost solutions $\Phi_{\pm}(x,t;\lambda)$ in Eq.\,\eqref{vo3} and $\widetilde{\Phi}_{\pm}(x,t;\lambda)$ in Eq.\,\eqref{vo2} are well-defined for $\lambda\in \mathbb{R}\backslash\{0\}$. Specifically, for every $x\in\mathbb{R}$, $\Phi_{+1}(x,t;\cdot),\widetilde{\Phi}_{+1}(x,t;\cdot)$ and $\Phi_{-2}(x,t;\cdot),\widetilde{\Phi}_{-2}(x,t;\cdot)$ can be analytically continued onto the regions $\mathbb{C}^{+}$, while $\Phi_{-1}(x,t;\cdot)$, $\widetilde{\Phi}_{-1}(x,t;\cdot)$ and $\Phi_{+2}(x,t;\cdot),\widetilde{\Phi}_{+2}(x,t;\cdot)$ can be analytically continued onto $\mathbb{C}^{-}$ with bounded limits as $\lambda\to0$ and $\lambda\to\infty$, as given by Eqs.\,\eqref{lim2}, \eqref{lim1} and \eqref{connect3}.
	\end{theorem}
	\subsection{Scattering matrix}
	In the following, we introduce the scattering matrix. 	Given that $\mathbf{Q}_1(x,t)$ and $\mathbf{Q}_2(x,t)$ are traceless, by Abel's Theorem, we arrive at  \begin{equation}\partial_{x}\operatorname*{det}(\Phi_{\pm}(x,t;\lambda))=0,\quad\lambda\in \mathbb{R}\backslash\{0\},\end{equation} which signifies that the determinant of $\Phi_{\pm}(x,t;\lambda)$ is independent of the spatial variable $x$. Consequently, we evaluate the determinant of $\Phi_{\pm}(x,t;\lambda)$ as $x\to\pm\infty$. Eq.\,\eqref{par3} implies:
	\begin{equation}\label{det2}\det(\Phi_{\pm}(x,t;\lambda))=1,\quad\lambda\in\mathbb{R}\backslash\{0\}.\end{equation}
	It is clear that the scattering matrix $\mathbf{S}(\lambda)$ satisfies the following relationship:
	\begin{equation}\label{sg2}
		\Phi_{-}(x,t;\lambda)=\Phi_{+}(x,t;\lambda)\mathrm{e}^{\mathrm{i}\Theta(x,t;\lambda)\hat{\sigma}_3}\mathbf{S}(\lambda),\quad\lambda\in \mathbb{R}\backslash\{0\}.
	\end{equation}
where $\Theta(x,t;\lambda)=\frac14(\lambda-\lambda^{-1})x+\frac14(\lambda+\lambda^{-1})t.$ 	
	To determine the reflection coefficient, which is pivotal for constructing the jump matrix, we revert to the analysis of the original Jost solutions $\{\phi_{+},\phi_{-}\}$. By examining the relationships among these Jost solutions, we derive the relationships between the elements of the scattering matrix and subsequently ascertain the reflection coefficient. It is imperative to investigate the symmetries of the scattering problem and
the scattering problem exhibits symmetry under the involution:\,$\zeta\to\bar{\zeta}$.
	\begin{prop}\label{sy}
		Considering the symmetries of the original Lax pair \eqref{lax}, we get the following symmetry relations:
		\begin{equation}\label{sym}
			\phi_\pm(x,t;\zeta)=\sigma_2\overline{\phi_\pm(x,t;\bar{\zeta})}\sigma_2,\quad\zeta\in\mathbb{R}\cup i\mathbb{R}\backslash\{0\}.
		\end{equation}	\end{prop}	Based on the relationship between $\Phi(x,t;\lambda)$ and $\phi(x,t;\zeta)$ as described in Eqs.\,\eqref{m} and \eqref{m3}, along with the analyticity of $\Phi(x,t;\lambda)$ established in Theorem \ref{th3}, we conclude that the symmetry relations \eqref{sym} can be analytically extended to broader domains.
		\begin{corollary}The columns of $\phi_{\pm}(x,t;\zeta)$ exhibit the following symmetry relations:	
			\begin{equation}\label{schw}\begin{aligned}
					&\phi_{+1}(x,t;\zeta)=\mathrm{i}\sigma_{2}\overline{\phi_{+2}(x,t;\bar{\zeta})},\quad\zeta\in \mathbb{C}^{++}\cup\mathbb{C}^{--}\cup\mathbb{R}\cup \mathrm{i}\mathbb{R}\backslash\{0\},\\
					&\phi_{-1}(x,t;\zeta)=\mathrm{i}\sigma_{2}\overline{\phi_{-2}(x,t;\bar{\zeta})},\quad\zeta\in \mathbb{C}^{+-}\cup\mathbb{C}^{-+}\cup\mathbb{R}\cup \mathrm{i}\mathbb{R}\backslash\{0\}.
					\end{aligned}\end{equation}
			\end{corollary}		
We consider the symmetries of the transformed Jost solutions $\Phi_{\pm}(x,t;\lambda)$ under the involution:\,$\lambda \to \bar{\lambda}$.
\begin{corollary}
	 According to the proposition\,\ref{sy}, we get the following symmetry relations:
	\begin{equation}
		\Phi_{\pm}(x,t;\lambda)=\mathbf{T}(x,t;\lambda)\sigma_2\mathbf{T}^{-1}(x,t;\lambda)\overline{\Phi_{\pm}(x,t;\bar{\lambda})}\sigma_2\left(\begin{array}{cc}\lambda^{\frac12}&0\\0&\lambda^{-\frac12}\end{array}\right),\quad\lambda\in\mathbb{R}\backslash\{0\},
	\end{equation}
	the symmetry relations can be analytically extended to broader domains:
	\begin{equation}\label{phi}\begin{aligned}
			&\Phi_{+1}(x,t;\lambda)=\mathrm{i}\lambda^{\frac12}\mathbf{T}(x,t;\lambda)\sigma_2\mathbf{T}^{-1}(x,t;\lambda)\overline{\Phi_{+2}(x,t;\bar{\lambda})},\quad\lambda\in \mathbb{C}^{+}\cup\mathbb{R}\backslash\{0\},\\
			&\Phi_{-1}(x,t;\lambda)=\mathrm{i}\lambda^{\frac12}\mathbf{T}(x,t;\lambda)\sigma_2\mathbf{T}^{-1}(x,t;\lambda)\overline{\Phi_{-2}(x,t;\bar{\lambda})},\quad\lambda\in \mathbb{C}^{-}\cup\mathbb{R}\backslash\{0\}.
	\end{aligned}\end{equation}
Moreover, we can also obtain the symmetry relations of $\mathbf{S}(\lambda)$:
\begin{equation}
	\mathbf{S}(\lambda)=\left(\begin{array}{cc}\lambda^{-\frac12}&0\\0&\lambda^{\frac12}\end{array}\right)\sigma_2\overline{\mathbf{S}(\bar{\lambda})}\sigma_2\left(\begin{array}{cc}\lambda^{\frac12}&0\\0&\lambda^{-\frac12}\end{array}\right),\quad\lambda\in\mathbb{R}\backslash\{0\}.
\end{equation}
\end{corollary}
	Consequently, the relationships among the elements of the scattering matrix $\mathbf{S}(\lambda)$ are derived as follows:
\begin{equation}
\mathbf{S}_{11}(\lambda)=\overline{\mathbf{S}_{22}(\bar{\lambda})},\quad\mathbf{S}_{21}(\lambda)=-\lambda\,\overline{\mathbf{S}_{12}(\bar{\lambda})},\quad\lambda\in\mathbb{R}\backslash\{0\}.
\end{equation}	
	 From these relations, the scattering matrix $\mathbf{S}(\lambda)$ can be written in the form
	\begin{equation}\label{sgl}\begin{aligned}
			\mathbf{S}(\lambda)=\left(\begin{array}{cc}\overline{\alpha(\bar{\lambda})}&{\beta}(\lambda)\\-\lambda\overline{\beta(\bar{\lambda})}&\alpha(\lambda)\end{array}\right),
			\quad\lambda\in \mathbb{R}\backslash\{0\}.
	\end{aligned}\end{equation}
Taking the determinants of both sides of Eq.\,\eqref{sg2} results in the following expression:
\begin{equation}
	|\alpha(\lambda)|^2+\lambda|\beta(\lambda)|^2=1,
	\quad\lambda\in \mathbb{R}\backslash\{0\}.
\end{equation}
	Moreover, the integral representations of ${\alpha}(\lambda)$ and ${\beta}(\lambda)$ can be derived from Eqs.\,\eqref{vo3} and \eqref{sg2} as expressed below:
	\begin{equation}\begin{aligned}
			{\alpha}(\lambda)=1+\int_{-\infty}^{+\infty}\Big[&\big(v_{x}+\frac{\mathrm{i}}{2}(u+|u|^2v-\lambda^{-1}v-\lambda^{-1}\bar{u} v^2)\big)\Phi_{-12}\\
			&+\big(\frac{\mathrm{i}}{2}\lambda^{-1}\bar{u}v-\frac{\mathrm{i}}{4}(|u|^{2}+|v|^{2})\big)\Phi_{-22}\Big](x,0;\lambda)\mathrm{d}x,\end{aligned}\end{equation} and
	\begin{equation}\begin{aligned}	{\beta}(\lambda)=\int_{-\infty}^{+\infty}\left[\mathrm{e}^{-2\mathrm{i}\Theta}\Big(\big(\frac{\mathrm{i}}{4}(|u|^{2}+|v|^{2})-\frac{\mathrm{i}}{2}\lambda^{-1}\bar{u} v\big)\Phi_{-12}
			+\big(\frac{\mathrm{i}}{2}\lambda^{-1}\bar{u}-\frac{\mathrm{i}}{2}\bar{v}\big)\Phi_{-22}\Big)\right](x,0;\lambda)\mathrm{d}x.
	\end{aligned}\end{equation}
Assuming $u(\cdot,0), v(\cdot,0)\in L^1(\mathbb{R})\cap L^{\infty}(\mathbb{R})$, and $v_{x}(\cdot,0)\in L^1(\mathbb{R})$, the functions ${\alpha}(\lambda)$ and ${\beta}(\lambda)$ are well-defined for $\lambda\in \mathbb{R}\backslash\{0\}$. Furthermore, ${\alpha}(\cdot)$ can be analytically continued onto $\mathbb{C}^{+}$.
	According to Eqs.\,\eqref{det2} and \eqref{sg2}, ${\alpha}(\lambda)$ and ${\beta}(\lambda)$ can also be expressed as Wronskians of $\Phi_{\pm}(x,t;\lambda)$. Specifically,
	\begin{equation}\label{wrons}
		\begin{aligned}
			&{\alpha}(\lambda)=\det\left(\Phi_{+1}(x,t;\lambda)\mathrm{e}^{\mathrm{i}\Theta(x,t;\lambda)},\Phi_{-2}(x,t;\lambda)\mathrm{e}^{-\mathrm{i}\Theta(x,t;\lambda)}\right),\quad\lambda\in \mathbb{C}^+\cup\mathbb{R}\backslash\{0\},\\
			&{\beta}(\lambda)=\det\left(\Phi_{-2}(x,t;\lambda)\mathrm{e}^{-\mathrm{i}\Theta(x,t;\lambda)},\Phi_{+2}(x,t;\lambda)\mathrm{e}^{\mathrm{i}\Theta(x,t;\lambda)}\right),\quad\lambda\in \mathbb{R}\backslash\{0\}.
	\end{aligned}\end{equation}

By substituting the Wentzel--Kramers--Brillouin expansions of the columns of the modified Jost solutions $\Phi_{\pm}(x,t;\lambda)$ into Eq.\,\eqref{par3} and collecting the terms of order $O(\lambda^j)$, we obtain the following asymptotic results:
\begin{prop}
	As $\lambda\to\infty$ within the corresponding domains of analyticity,
	\begin{equation}
		\lim_{\lambda\to\infty}\Phi_{\pm}(x,t;\lambda)=\left(\begin{array}{cc}\Phi_{\pm1}^{\infty}(x,t)&0\\0&\Phi_{\pm2}^{\infty}(x,t)\end{array}\right).
	\end{equation}
	Furthermore, the functions $\alpha(\lambda)$ and $\beta(\lambda)$ exhibit the following asymptotic behavior as $\lambda\to0$ and $\lambda\to\infty$ in the domains of their analyticity:
\begin{equation}\begin{aligned}
\lim_{\lambda\to0}\alpha(\lambda)&=\mathrm{e}^{\frac{\mathrm{i}}{4}\int_{\mathbb{R}}\left(|u|^2+|v|^2\right)\mathrm{d}y},\qquad\lim_{\lambda\to0}\beta(\lambda)=0,\\
\lim_{\lambda\to\infty}\alpha(\lambda)&=\mathrm{e}^{-\frac{\mathrm{i}}{4}\int_{\mathbb{R}}\left(|u|^2+|v|^2\right)\mathrm{d}y},\,\quad\lim_{\lambda\to\infty}\beta(\lambda)=0.
\end{aligned}\end{equation}	
	\end{prop}
Define the reflection coefficient as
	\begin{equation}
		\begin{aligned}
			{\gamma}(\lambda)=\frac{{\beta}(\lambda)}{{\alpha}(\lambda)},\quad\lambda\in\mathbb{R}\backslash\{0\},
		\end{aligned}
	\end{equation}
and denote $\frac{1}{\alpha(\lambda)}$ as the transmission coefficient. It is noteworthy that ${\gamma}(\lambda)\to0$ as $\lambda\to0$ or $\lambda\to\infty$.
	
	\subsection{Discrete spectrum}
	For a (vector-valued) function $\mathbf{f}(x,t;\lambda)$, let $\mathbf{f}^{(n)}(x,t;\lambda)$ denote the $n$-th partial derivative with respect to $\lambda$, i.e.,\,$\mathbf{f}^{(n)}(x,t;\lambda)=\partial_{\lambda}^n\mathbf{f}(x,t;\lambda)$. Moreover, $\mathbf{f}^{(n)}(x,t;\lambda_0)$ refers to the $n$-th derivative evaluated at $\lambda=\lambda_0$, thus $\mathbf{f}^{(n)}(x,t;\lambda_0)=\mathbf{f}^{(n)}(x,t;\lambda)|_{\lambda=\lambda_0}$. For simplicity, we sometimes omit the variables $x$ and $t$.
	\begin{prop}\label{b}
		Suppose that $\lambda_0$ is a zero of $\alpha(\lambda)$ with multiplicity $m+1$, there exist $m+1$ complex-valued constants $b_0,b_1,\ldots,b_m$ with $b_0\neq0$ such that, for every $n\in\{0,\ldots,m\}$, the following holds:
		\begin{equation}\frac{(\Phi_{-2}\mathrm{e}^{-\mathrm{i}\Theta})^{(n)}(x,t;\lambda_0)}{n!}=\sum_{\overset{j+k=n}{j,k\geqslant0}}\frac{b_j(\Phi_{+1}\mathrm{e}^{\mathrm{i}\Theta})^{(k)}(x,t;\lambda_0)}{j!k!}.\end{equation}
	\end{prop}
	\begin{proof} This proposition can be obtained from \eqref{wrons} and Proposition 2.4 of Ref.\cite{AL}.\end{proof}
	\begin{corollary}\label{co}
		Suppose that $\lambda_0$ is a zero of $\alpha(\lambda)$ with multiplicity $m+1$, thus for each $n\in\{0,\dots,m\}$,
		\begin{equation}\frac{\Phi_{-2}^{(n)}(x,t;\lambda_0)}{n!}=\sum_{\overset{j+k+l=n}{j,k,l\geqslant0}}\frac{b_j(\mathrm{e}^{2\mathrm{i}\Theta})^{(k)}(x,t;\lambda_0)\Phi_{+1}^{(l)}(x,t;\lambda_0)}{j!k!l!},\end{equation}
		where $b_0,b_1,\ldots,b_m$ are given in Proposition $\ref{b}$.
	\end{corollary}
	\begin{proof}
		It follows from Proposition $\ref{b}$ that
		\begin{equation}\begin{aligned}
				\frac{\Phi_{-2}^{(n)}(\lambda_0)}{n!}&=\frac{(\Phi_{-2}\mathrm{e}^{-\mathrm{i}\Theta}\mathrm{e}^{\mathrm{i}\Theta})^{(n)}(\lambda_0)}{n!}=\sum_{\overset{r+s=n}{r,s\geqslant0}}\frac{(\Phi_{-2}\mathrm{e}^{-\mathrm{i}\Theta})^{(r)}(\lambda_0)(\mathrm{e}^{\mathrm{i}\Theta})^{(s)}(\lambda_0)}{r!s!}\\
				&=\sum_{\overset{r+s=n}{r,s\geqslant0}}\sum_{\overset{j+m=r}{j,m\geqslant0}}\frac{b_j(\Phi_{+1}\mathrm{e}^{\mathrm{i}\Theta})^{(m)}(\lambda_0)(\mathrm{e}^{\mathrm{i}\Theta})^{(s)}(\lambda_0)}{j!m!s!}\\
				&=\sum_{\overset{j+l+h+s=n}{j,l,h,s\geqslant0}}\frac{b_j{\Phi}^{(l)}_{+1}(\lambda_0)(\mathrm{e}^{\mathrm{i}\Theta})^{(h)}(\lambda_0)(\mathrm{e}^{\mathrm{i}\Theta})^{(s)}(\lambda_0)}{j!l!h!s!}\\
				&=\sum_{\overset{j+l+k=n}{j,l,k\geqslant0}}\sum_{\overset{h+s=k}{h,s\geqslant0}}\frac{b_j{\Phi}^{(l)}_{+1}(\lambda_0)}{j!l!}\frac{(\mathrm{e}^{\mathrm{i}\Theta})^{(h)}(\lambda_0)(\mathrm{e}^{\mathrm{i}\Theta})^{(s)}(\lambda_0)}{h!s!}\\
				&=\sum_{\overset{j+k+l=n}{j,k,l\geqslant0}}\frac{b_j(\mathrm{e}^{2\mathrm{i}\Theta})^{(k)}(\lambda_0){\Phi}^{(l)}_{+1}(\lambda_0)}{j!k!l!}.
		\end{aligned}\end{equation}
	\end{proof}
	\begin{prop}\label{mu}
		Suppose that $\lambda_0$ is a zero of $\alpha(\lambda)$ with multiplicity $m+1$, then $\bar{\lambda}_0$ is a zero of $\overline{\alpha(\bar{\lambda})}$ with multiplicity $m+1$. For each $n\in\{0,\dots,m\}$,
		\begin{equation}\frac{\Phi_{-1}^{(n)}(x,t;\bar{\lambda}_0)}{n!}=\sum_{\overset{j+k+l=n}{j,k,l\geqslant0}}-\frac{\bar{b}_j(\mathrm{e}^{-2\mathrm{i}\Theta(x,t;\lambda)})^{(k)}|_{\lambda=\bar{\lambda}_0}\big(\lambda\Phi_{+2}(x,t;\lambda)\big)^{(l)}|_{\lambda=\bar{\lambda}_0}}{j!k!l!},\end{equation}
		where $b_0,b_1,\ldots,b_m$ are given in Proposition $\ref{b}$.
	\end{prop}
	\begin{proof}
		It follows from  Corollary $\ref{co}$ that
		\begin{equation}\begin{aligned}
				&\frac{\Phi_{-1}^{(n)}(\bar{\lambda}_0)}{n!}=\frac{\big(\mathrm{i}\lambda^{\frac12}\mathbf{T}(\lambda)\sigma_2\mathbf{T}^{-1}(\lambda)\overline{\Phi_{-2}(\bar{\lambda})}\big)^{(n)}|_{\lambda=\bar{\lambda}_0}}{n!}\\
				=&\sum_{\overset{r+s=n}{r,s\geqslant0}}\frac{\mathrm{i}\big(\lambda^{\frac12}\mathbf{T}(\lambda)\sigma_2\mathbf{T}^{-1}(\lambda)\big)^{(r)}|_{\lambda=\bar{\lambda}_0}\big(\overline{\Phi_{-2}(\bar{\lambda})}\big)^{(s)}|_{\lambda=\bar{\lambda}_0}}{r!s!}\\
				=&\sum_{\overset{r+s=n}{r,s\geqslant0}}\sum_{\overset{j+k+m=s}{j,k,m\geqslant0}}\frac{\mathrm{i}\big(\lambda^{\frac12}\mathbf{T}(\lambda)\sigma_2\mathbf{T}^{-1}(\lambda)\big)^{(r)}|_{\lambda=\bar{\lambda}_0}\bar{b}_j(\mathrm{e}^{-2\mathrm{i}\Theta(\lambda)})^{(k)}|_{\lambda=\bar{\lambda}_0}\big(\overline{\Phi_{+1}(\bar{\lambda})}\big)^{(m)}|_{\lambda=\bar{\lambda}_0}}{r!j!k!m!}\\
				=&\sum_{\overset{r+j+k+m=n}{r,j,k,m\geqslant0}}\frac{\mathrm{i}\big(\lambda^{\frac12}\mathbf{T}(\lambda)\sigma_2\mathbf{T}^{-1}(\lambda)\big)^{(r)}|_{\lambda=\bar{\lambda}_0}\bar{b}_j(\mathrm{e}^{-2\mathrm{i}\Theta(\lambda)})^{(k)}|_{\lambda=\bar{\lambda}_0}\big(\mathrm{i}\lambda^{\frac12}\mathbf{T}(\lambda)\sigma_2\mathbf{T}^{-1}(\lambda)\Phi_{+2}(\lambda)\big)^{(m)}|_{\lambda=\bar{\lambda}_0}}{r!j!k!m!}\\
				=&\sum_{\overset{j+k+l=n}{j,k,l\geqslant0}}\sum_{\overset{r+m=l}{r,m\geqslant0}}-\frac{\big(\lambda^{\frac12}\mathbf{T}(\lambda)\sigma_2\mathbf{T}^{-1}(\lambda)\big)^{(r)}|_{\lambda=\bar{\lambda}_0}\big(\lambda^{\frac12}\mathbf{T}(\lambda)\sigma_2\mathbf{T}^{-1}(\lambda)\Phi_{+2}(\lambda)\big)^{(m)}|_{\lambda=\bar{\lambda}_0}}{r!m!}\frac{\bar{b}_j(\mathrm{e}^{-2\mathrm{i}\Theta(\lambda)})^{(k)}|_{\lambda=\bar{\lambda}_0}}{j!k!}\\
				=&\sum_{\overset{j+k+l=n}{j,k,l\geqslant0}}-\frac{\bar{b}_j(\mathrm{e}^{-2\mathrm{i}\Theta(\lambda)})^{(k)}|_{\lambda=\bar{\lambda}_0}\big(\lambda\Phi_{+2}(\lambda)\big)^{(l)}|_{\lambda=\bar{\lambda}_0}}{j!k!l!}.
		\end{aligned}\end{equation}
	\end{proof}
	Given that $\lambda_0$ is a zero of ${\alpha}(\lambda)$ with multiplicity $m+1$. The Laurent series expansion of the transmission coefficient  $\frac{1}{{\alpha}(\lambda)}$ around $\lambda=\lambda_0$ can be derived as follows:
	\begin{equation}\frac1{{\alpha}(\lambda)}=\frac{\alpha_{-m-1}}{(\lambda-\lambda_0)^{m+1}}+\frac{\alpha_{-m}}{(\lambda-\lambda_0)^m}+\cdots+\frac{\alpha_{-1}}{\lambda-\lambda_0}+O(1),\quad\lambda\to\lambda_0,\end{equation}
	where $\alpha_{-m-1}\neq0$ and $
	\alpha_{-n-1}=\frac{\widetilde{\alpha}^{(m-n)}(\lambda_0)}{(m-n)!}$, $\widetilde{\alpha}(\lambda)=\frac{(\lambda-\lambda_0)^{m+1}}{\alpha(\lambda)}, n=0,\ldots,m.$ Combining with Corollary $\ref{co}$, we obtain that for each $n\in\{0,\ldots,m\},$
	\begin{equation}\underset{\lambda_0}{\operatorname*{Res}}\frac{(\lambda-\lambda_0)^n \Phi_{-2}(x,t;\lambda)}{\alpha(\lambda)}=\sum_{\overset{j+k+l+s=m-n}{j,k,l,s\geqslant0}}\frac{\widetilde{\alpha}^{(j)}(\lambda_0)b_k(\mathrm{e}^{2\mathrm{i}\Theta})^{(l)}(x,t;\lambda_0)\Phi_{+1}^{(s)}(x,t;\lambda_0)}{j!k!l!s!}.\end{equation}
	To simplify the above residue, we introduce a polynomial of degree at most $m$:
	\begin{equation}\label{f}f_0(\lambda)=\sum_{h=0}^m\sum_{\substack{j+k=h\\j,k\geqslant0}}\frac{\widetilde{\alpha}^{(j)}(\lambda_0)b_k}{j!k!}(\lambda-\lambda_0)^h,\end{equation}
	with $f_0(\lambda_0)\neq0$. Therefore,
	\begin{equation}\label{resd}
		\begin{aligned}
			\underset{\lambda_0}{\operatorname*{Res}}\frac{(\lambda-\lambda_0)^n\Phi_{-2}(x,t;\lambda)}{\alpha(\lambda)}=&\sum_{\overset{h+l+s=m-n}{h,l,s\geqslant0}}\frac{f_0^{(h)}(\lambda_0)(\mathrm{e}^{2\mathrm{i}\Theta})^{(l)}(x,t;\lambda_0)\Phi_{+1}^{(s)}(x,t;\lambda_0)}{h!l!s!}\\
			=&\frac{[f_0(\lambda)\mathrm{e}^{2\mathrm{i}\Theta(x,t;\lambda)}\Phi_{+1}(x,t;\lambda)]^{(m-n)}|_{\lambda=\lambda_0}}{(m-n)!}.
		\end{aligned}
	\end{equation}
	We denote $f_0(\lambda_0),\dots,f_0^{(m)}(\lambda_0)$ as the residue constants associated with the discrete spectrum $\lambda_0$.
	Given that $\bar{\lambda}_0$ is a zero of $\overline{\alpha(\bar{\lambda})}$ with multiplicity $m+1$, and we get the results according to Proposition $\ref{mu}$,
	\begin{equation}\label{res2}
		\begin{aligned}
			&\underset{\bar{\lambda}_0}{\operatorname*{Res}}\frac{(\lambda-\bar{\lambda}_0)^n\Phi_{-1}(x,t;\lambda)}{\overline{\alpha(\bar{\lambda})}}\\
			=&\sum_{\overset{j+k+l+s=m-n}{j,k,l,s\geqslant0}}-\frac{\big(\overline{\widetilde{\alpha}(\bar{\lambda})}\big)^{(j)}|_{\lambda=\bar{\lambda}_0}\bar{b}_k\big(\mathrm{e}^{-2\mathrm{i}\Theta(x,t;\lambda)}\big)^{(l)}|_{\lambda=\bar{\lambda}_0}\big(\lambda\Phi_{+2}(x,t;\lambda)\big)^{(s)}|_{\lambda=\bar{\lambda}_0}}{j!k!l!s!}\\
			=&\sum_{\overset{h+l+s=m-n}{h,l,s\geqslant0}}-\frac{\big(\overline{f_0(\bar{\lambda})}\big)^{(h)}|_{\lambda=\bar{\lambda}_0}\big(\mathrm{e}^{-2\mathrm{i}\Theta(x,t;\lambda)}\big)^{(l)}|_{\lambda=\bar{\lambda}_0}\big(\lambda\Phi_{+2}(x,t;\lambda)\big)^{(s)}|_{\lambda=\bar{\lambda}_0}}{h!l!s!}\\
			=&-\frac{[\lambda\overline{f_0(\bar{\lambda})} \mathrm{e}^{-2\mathrm{i}\Theta(x,t;\lambda)}\Phi_{+2}(x,t;\lambda)]^{(m-n)}|_{\lambda=\bar{\lambda}_0}}{(m-n)!}.
		\end{aligned}
	\end{equation}
	\begin{assume}\label{alpha}
		Given that $\alpha(\lambda)$ possesses $N$ distinct zeros $\lambda_1,\ldots,\lambda_N\in\mathbb{C}^+$, each zero has a multiplicity of $m_1+1,\ldots,m_N+1$, respectively. None of these zeros lie on the real axis.
	\end{assume}
	\begin{prop}\label{res}
		Assuming $\lambda_1,\ldots,\lambda_N$ are as stated in Assumption $\ref{alpha}$, there exists a unique polynomial $f(\lambda)$ of degree less than $\mathcal{N}=\sum_{k=1}^{N}(m_k+1)$ with $f(\lambda_k)\neq0$ such that, for $k=1,\ldots,N$ and $n_k=0,\ldots,m_k$, the following residue conditions are satisfied:
		\begin{alignat}{2}
			&\underset{\lambda_k}{\operatorname*{Res}}\frac{(\lambda-\lambda_k)^{n_k}\Phi_{-2}(x,t;\lambda)}{\alpha(\lambda)}=\frac{[f(\lambda)\mathrm{e}^{2\mathrm{i}\Theta(x,t;\lambda)}\Phi_{+1}(x,t;\lambda)]^{(m_k-n_k)}|_{\lambda=\lambda_k}}{(m_k-n_k)!},\label{prop:res}\\
			&\underset{\bar{\lambda}_k}{\operatorname*{Res}}\frac{(\lambda-\bar{\lambda}_k)^{n_k}\Phi_{-1}(x,t;\lambda)}{\overline{\alpha(\bar{\lambda})}}=-\frac{[\lambda\overline{f(\bar{\lambda})}\mathrm{e}^{-2\mathrm{i}\Theta(x,t;\lambda)}\Phi_{+2}(x,t;\lambda)]^{(m_k-n_k)}|_{\lambda=\bar{\lambda}_k}}{(m_k-n_k)!}.\label{prop:res2}
		\end{alignat}
	\end{prop}
	\begin{proof}Similar to Eqs.\,\eqref{resd} and \eqref{res2}, for each $k\in\{1,\dots,N\}$, there exists a polynomial $f_k(\lambda)$ of degree at most $m_k$ with $f_k(\lambda_k)\neq0$ such that
			\begin{alignat}{2}\label{res3}
				&\underset{\lambda_k}{\operatorname*{Res}}\frac{(\lambda-\lambda_k)^n\Phi_{-2}(x,t;\lambda)}{{\alpha}(\lambda)}
				=\frac{[f_k(\lambda)\mathrm{e}^{2\mathrm{i}\Theta(x,t;\lambda)}\Phi_{+1}(x,t;\lambda)]^{(m_k-n_k)}|_{\lambda=\lambda_k}}{(m_k-n_k)!},\\
				&\underset{\bar{\lambda}_k}{\operatorname*{Res}}\frac{(\lambda-\bar{\lambda}_k)^{n_k}\Phi_{-1}(x,t;\lambda)}{\overline{\alpha(\bar{\lambda})}}=-\frac{[\lambda\overline{f_k(\bar{\lambda})}\mathrm{e}^{-2\mathrm{i}\Theta(x,t;\lambda)}\Phi_{+2}(x,t;\lambda)]^{(m_k-n_k)}|_{\lambda=\bar{\lambda}_k}}{(m_k-n_k)!},
		\end{alignat}
		where $n_k=0,\dots,m_k$. By the Hermite interpolation formula, there exists a unique polynomial
		$f(\lambda)$ of degree less than $\mathcal{N}$ such that
		\begin{equation}\label{h}\begin{cases}f^{(n_1)}(\lambda_1)=f^{(n_1)}_1(\lambda_1),&\quad n_1=0,\ldots,m_1,\\\quad\quad\quad\quad\vdots\\f^{(n_N)}(\lambda_N)=f^{(n_N)}_N(\lambda_N),&\quad n_N=0,\ldots,m_N.\end{cases}\end{equation}
		This completes the proof of Eqs.\,\eqref{prop:res} and \eqref{prop:res2}.
	\end{proof}
	\begin{remark}
		Similar to Eq.\,\eqref{f}, the polynomial $f(\lambda)$ presented in Proposition $\ref{res}$ is uniquely determined by $\alpha(\lambda),\{\lambda_k,m_k\}_{k=1}^N$ and a set of constants $\{b_{k,0},\dots,b_{k,m_k}\}_{k=1}^N.$ Assume that $g(\lambda)$ is analytic at $\lambda_1, \dots\lambda_N$,  the substitution of $f(\lambda)$ for $f(\lambda)+g(\lambda)\prod_{k=1}^{N}(\lambda-\lambda_{k})^{m_{k}+1}$ still preserves the validity of Eqs.\,\eqref{h}. Consequently, Proposition $\ref{res}$ can be revised to state that ``there exists a function ${f}(\lambda)$ that is analytic and nonzero at $\lambda_1,\dots,\lambda_N$, satisfying Eqs.\,\eqref{prop:res} and \eqref{prop:res2}''.
	\end{remark}
	
	\section{Inverse problem}\label{sec:inverse}
	In Section $\ref{sec:dir}$, we have established the direct scattering map:
	\begin{equation}\mathcal{D}:\{u(x,0), v(x,0)\}\mapsto\left\{{\gamma}(\lambda),\left\{\lambda_k,m_k,\left\{f^{(n_k)}(\lambda_k)\right\}_{n_k=0}^{m_k}\right\}_{k=1}^N\right\}.\end{equation}
	In the following, we will now turn our attention to the inverse scattering map:
	\begin{equation}\mathcal{I}:\left\{{\gamma}(\lambda),\left\{\lambda_k,m_k,\left\{f^{(n_k)}(\lambda_k)\right\}_{n_k=0}^{m_k}\right\}_{k=1}^N\right\}\mapsto \{u(x,t),v(x,t)\},\end{equation}
	which can be characterised by a $2\times2$ matrix RH problem.
	\subsection{Riemann--Hilbert problem}
	Consider a piecewise meromorphic function $\widetilde{\mathbf{M}}(x,t;\lambda)$ defined by
	\begin{equation}
		\begin{aligned}
			&\widetilde{\mathbf{M}}(x,t;\lambda)=\left(\Phi_{+1}(x,t;\lambda),\frac{\Phi_{-2}(x,t;\lambda)}{\alpha(\lambda)}\right),\quad\lambda\in\mathbb{C}^+,\\
			&\widetilde{\mathbf{M}}(x,t;\lambda)=\left(\frac{\Phi_{-1}(x,t;\lambda)}{\overline{\alpha(\bar{\lambda})}},\Phi_{+2}(x,t;\lambda)\right),\quad\lambda\in\mathbb{C}^-.
		\end{aligned}
	\end{equation}
	The following behavior of $\mathbf{M}(x,t;\lambda)$ holds as $\lambda\to\infty$ within its analytic domains,
	\begin{equation}
		\begin{aligned}
			\widetilde{\mathbf{M}}(x,t;\lambda)\to\left(\begin{array}{cc}\Phi_{+1}^{\infty}(x,t)&0\\0&\Phi_{+2}^{\infty}(x,t)\end{array}\right)=\widetilde{\mathbf{M}}^\infty(x,t).
		\end{aligned}
	\end{equation}
	To normalize the boundary conditions, we define
	\begin{equation}
		\mathbf{M}(x,t;\lambda)=\left(\widetilde{\mathbf{M}}^\infty(x,t)\right)^{-1}\widetilde{\mathbf{M}}(x,t;\lambda),\quad \lambda\in\mathbb{C}\backslash\{0\}.
	\end{equation}
	From this, we can formulate a $2\times 2$ matrix RH problem.
	\begin{problem}\label{pro}
		Find a $2\times2$ matrix-valued function ${\mathbf{M}}(x,t;\lambda)$ that satisfies the following properties:
		\begin{itemize}
			\item \textbf{Analyticity:} $\mathbf{M}(x,t;\lambda)$ is analytic in $\lambda$ for $\lambda\in\mathbb{C}\backslash\left(\mathbb{R}\cup\{\lambda_1,\ldots,\lambda_N,\bar{\lambda}_1,\ldots,\bar{\lambda}_N\}\right)$;
			\item \textbf{Residues:} For each $k=1,\ldots,N$, $\mathbf{M}(x,t;\lambda)$ has multiple-poles of order $m_k+1$  at $\lambda=\lambda_k$ and $\lambda=\bar{\lambda}_k$, with residues satisfy the following conditions:
\begin{alignat}{2}\label{res4}
					\underset{\lambda_k}{\operatorname*{Res}}(\lambda-\lambda_k)^{n_k}\mathbf{M}(x,t;\lambda)&=\left(\mathbf{0},\frac{[f(\lambda)\mathrm{e}^{2\mathrm{i}\Theta(x,t;\lambda)}\mathbf{M}_1(x,t;\lambda)]^{(m_k-n_k)}|_{\lambda=\lambda_k}}{(m_k-n_k)!}\right),\\
				\underset{\bar{\lambda}_k}{\operatorname*{Res}} (\lambda-\bar{\lambda}_k)^{n_k}\mathbf{M}(x,t;\lambda)&=\left(-\frac{[\lambda\overline{f(\bar{\lambda})}\mathrm{e}^{-2\mathrm{i}\Theta(x,t;\lambda)}\mathbf{M}_2(x,t;\lambda)]^{(m_k-n_k)}|_{\lambda=\bar{\lambda}_k}}{(m_k-n_k)!},\mathbf{0}\right),
\end{alignat}

			for $n_k=0,\ldots,m_k$;
			\item \textbf{Jump:} $\mathbf{M}(x,t;\lambda)$ exhibits a jump across the oriented contour $\mathbb{R}$ given by:
			\begin{equation}
				\mathbf{M}_{+}(x,t;\lambda)=\mathbf{M}_{-}(x,t;\lambda){\mathbf{J}}(x,t;\lambda),\quad\lambda\in\mathbb{R}\backslash\{0\},
			\end{equation}
			where $\mathbf{M}_{\pm}(x,t;\lambda)=\lim\limits_{\epsilon\downarrow0}\mathbf{M}(x,t;\lambda\pm\mathrm{i}\epsilon),$ and the jump matrix $\mathbf{J}(x,t;\lambda)$ is defined as:
			\begin{equation}\label{jump}\begin{aligned}
					\mathbf{J}(x,t;\lambda) =\left(\begin{array}{cc}1&\gamma(\lambda)\mathrm{e}^{2\mathrm{i}\Theta(x,t;\lambda)}\\\lambda\overline{\gamma(\bar{\lambda})}\mathrm{e}^{-2\mathrm{i}\Theta(x,t;\lambda)}&1+\lambda|\gamma(\lambda)|^2\end{array}\right);
			\end{aligned}\end{equation}
			\item \textbf{Normalization:} As $\lambda\in\mathbb{C}\backslash{\mathbb{R}}\to\infty$, the matrix $\mathbf{M}(x,t;\lambda)$ has the following asymptotic behavior:
			\begin{equation}\label{o}\mathbf{M}(x,t;\lambda)=\mathbf{I}+O(\lambda^{-1}).\end{equation}
		\end{itemize}
	\end{problem}
	For each $k$, let $D_k$ be a small disk centered at $\lambda_k$ with a sufficiently small radius such that it lies in $\mathbb{C}^+$ and is disjoint from all other disks as well as from $\{\hat{D}_k\}_{k=1}^N$, where $\hat{D}_k=\{\lambda\in\mathbb{C}^-|\bar{\lambda}\in D_k\}$. The boundaries $\{\partial{D}_k\}_{k=1}^N$ are oriented in a clockwise direction, while the boundaries $\{\partial\hat{D}_k\}_{k=1}^N$ are oriented in a counterclockwise direction, as illustrated in Figure \ref{fig1}.
	 \begin{figure}[!htb]\centering
	\begin{tikzpicture}[>=stealth, scale=0.5]
		\draw[->, fill=black] (5,0)--(26,0) node[right]{${Re} $};
		\fill[gray!50] (16.5,2) -- (18,2) arc (0:360:1.5) -- cycle;
		\draw[->](16.5,0.5) arc[start angle=270, end angle=90, radius=1.5];
		\draw[->] (18,2) arc[start angle=360, end angle=270, radius=1.5];
		\draw(18,2)arc(0:91:1.5);
		\node at (16.5,2) {$\cdot$};
		\node at  (16.5,1.5) {$\small\small{\lambda_k}$};
		\node at (17.3,2.4) {$\tiny{D_k}$};
		
				\fill[gray!50] (16.5,-2) -- (18,-2) arc (0:360:1.5) -- cycle;
			\draw[->](16.5,-0.5) arc[start angle=90, end angle=270, radius=1.5];
			\draw[->] (18,-2) arc[start angle=0, end angle=90, radius=1.5];
			\draw(16.5,-3.5)arc[start angle=270, end angle=360, radius=1.5];
			\node at (16.5,-2) {$\cdot$};
			\node at  (16.5,-2.6) {$\small\small{\bar{\lambda}_k}$};
			\node at (17.3,-1.6) {$\tiny{\hat D_k}$};
		
		\fill[gray!50] (9,3) -- (11,3) arc (0:360:2) -- cycle;
		\draw[->] (9,1) arc[start angle=270, end angle=90, radius=2];
		\draw[->] (11,3) arc[start angle=360, end angle=270, radius=2];
		\draw(11,3)arc(0:91:2);
		\node at (9,3) {$\cdot$};
		\node at  (9,2.5) {$\tiny{\lambda_1}$};
		\node at (10,3.6) {$\tiny{D_1}$};
		
		\fill[gray!50] (9,-3) -- (11,-3) arc (0:360:2) -- cycle;
		\draw[->] (9,-1) arc[start angle=90, end angle=270, radius=2];
		\draw(9,-5) arc[start angle=270, end angle=360, radius=2];
		\draw[->] (11,-3) arc[start angle=0, end angle=90, radius=2];
		\node at (9,-3) {$\cdot$};
		\node at  (9,-3.6) {$\tiny{\bar{\lambda}_1}$};
		\node at (10,-2.4) {$\tiny{\hat D_1}$};
		
		\fill[gray!50] (23,4) -- (24.8,4) arc (0:360:1.8) -- cycle;
		\draw[->] (23,2.2) arc[start angle=270, end angle=90, radius=1.8];
		\draw[->] (24.8,4) arc[start angle=360, end angle=270, radius=1.8];
		\draw(24.8,4)arc(0:91:1.8);
		\node at (23,4) {$\cdot$};
		\node at (23,3.5) {$\tiny{\lambda_N}$};
		\node at (24,4.5) {$\tiny{D_N}$};
		
			\fill[gray!50] (23,-4) -- (24.8,-4) arc (0:360:1.8) -- cycle;
			\draw[->] (23,-2.2) arc[start angle=90, end angle=270, radius=1.8];
			\draw[->] (24.8,-4) arc[start angle=0, end angle=90, radius=1.8];
			\draw (23,-5.8) arc[start angle=270, end angle=360, radius=1.8];
			\node at (23,-4) {$\cdot$};
			\node at (23,-4.6) {$\tiny{\bar{\lambda}_N}$};
			\node at (24,-3.5) {$\tiny{\hat D_N}$};
		\node at (13,-3) {$\cdots$};
		\node at (13,3) {$\cdots$};
		\node at (19.5,3) {$\cdots$};
		\node at (19.5,-3) {$\cdots$};
	\end{tikzpicture}
	\caption{ The regions $\{D_k\}_{k=1}^N$ and $\{\hat{D}_k\}_{k=1}^N$ (shaded).}
	\label{fig1}	
	\end{figure}
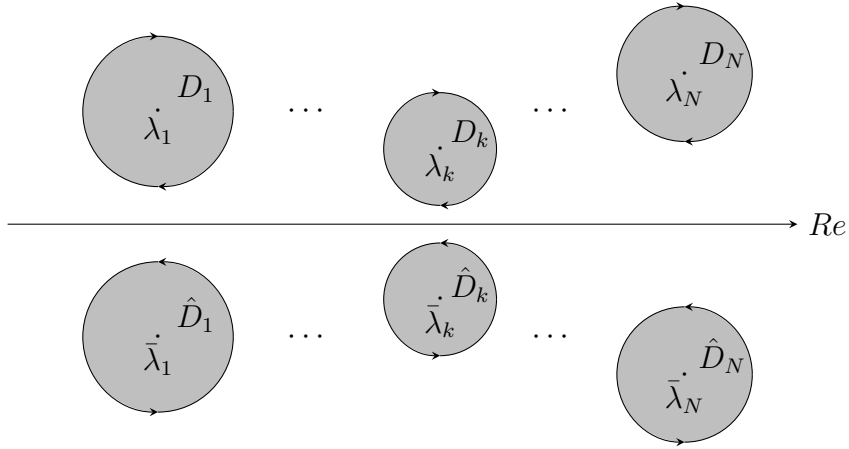
			
	We define a new matrix $\hat{\mathbf{M}}(x,t;\lambda)$ in terms of $\mathbf{M}(x,t;\lambda)$ in RH problem \ref{pro} by
	\begin{equation}\label{hatM}
		\hat{\mathbf{M}}(x,t;\lambda)=\begin{cases}
			\mathbf{M}(x,t;\lambda)\mathbf{P}_k(x,t;\lambda),\quad &\lambda\in D_k,\quad k=1,\ldots,N,\\
		\mathbf{M}(x,t;\lambda)\hat{\mathbf{P}}_k(x,t;\lambda),\quad &\lambda\in \hat{D}_k,\quad k=1,\ldots,N,\\
			\mathbf{M}(x,t;\lambda),\quad &\text{otherwise},
		\end{cases}
	\end{equation}where
	\begin{equation}\label{p}
		\mathbf{P}_k(x,t;\lambda)=\begin{pmatrix}
			1&-\frac{f(\lambda)\mathrm{e}^{2\mathrm{i}\Theta(x,t;\lambda)}}{(\lambda-\lambda_k)^{m_k+1}}\\
			0&1
		\end{pmatrix},\quad
		\hat{\mathbf{P}}_k(x,t;\lambda)=\begin{pmatrix}
			1&0\\
			\frac{\lambda \overline{f(\bar{\lambda})}\mathrm{e}^{-2\mathrm{i}\Theta(x,t;\lambda)}}{(\lambda-\bar{\lambda}_k)^{m_k+1}}&1
		\end{pmatrix}.
	\end{equation}
	It can be observed that the matrix $\mathbf{M}(x,t;\lambda)\mathbf{P}_k(x,t;\lambda)$ exhibits a removeable singularity at $\lambda_k$, while the matrix $\mathbf{M}(x,t;\lambda)\hat{\mathbf{P}}_k(x,t;\lambda)$ also demonstrates a removeable singularity at $\bar\lambda_k$. Indeed,
	\begin{equation}
		\begin{split}
			&\underset{\lambda_k}{\res} (\lambda-\lambda_k)^{n_k}\hat{\mathbf{M}}(x,t;\lambda)\\
			=&\underset{\lambda_k}{\res}(\lambda-\lambda_k)^{n_k}\left(\mathbf{M}_1(x,t;\lambda), \mathbf{M}_2(x,t;\lambda)-\frac{f(\lambda)\mathrm{e}^{2\mathrm{i}\Theta(x,t;\lambda)}}{(\lambda-\lambda_k)^{m_k+1}}\mathbf{M}_1(x,t;\lambda)\right)\\
			=&\left(\mathbf{0},\underset{\lambda_k}{\res}(\lambda-\lambda_k)^{n_k}\mathbf{M}_2(x,t;\lambda)-\underset{\lambda_k}{\res}\frac{f(\lambda)\mathrm{e}^{2\mathrm{i}\Theta(x,t;\lambda)}\mathbf{M}_1(x,t;\lambda)}{(\lambda-\lambda_k)^{m_k-n_k+1}}\right),
		\end{split}
	\end{equation}
	 this assertion can be seen by considering the residue condition \eqref{res4} and analyzing the Taylor series expansion of $f(\lambda)\mathrm{e}^{2\mathrm{i}\Theta(x,t;\lambda)}\mathbf{M}_1(x,t;\lambda)$ at $\lambda_k$. Consequently, for each $k$ and any $n_k$ satisfying $0\leqslant n_k\leqslant m_k$, we have $ \underset{\lambda_k}{\res} (\lambda-\lambda_k)^{n_k}\hat{\mathbf{M}}(x,t;\lambda)=\mathbf{0}$. By analogous reasoning, it can also be shown that the matrix $\mathbf{M}(x,t;\lambda)\hat{\mathbf{P}}_k(x,t;\lambda)$  has a removable singularity
	at $\bar{\lambda}_k$. Given the definition \eqref{hatM} of $\hat{\mathbf{M}}(x,t;\lambda)$, we conclude that all points in the set $\{\lambda_k,\bar\lambda_k\}_{k=1}^N$ are removable singularities of $\hat{\mathbf{M}}(x,t;\lambda)$.
	Therefore, we affirm that $\hat{\mathbf{M}}(x,t;\lambda)$ fulfills the conditions of an equivalent RH problem that is closely related
	to RH problem \ref{pro}, but with the residue conditions being substituted by jump conditions across small circles centered at the points $\{\lambda_k,\bar\lambda_k\}_{k=1}^N$.
	\begin{problem}\label{pro2}
		Find a $2\times2$ matrix-valued function $\hat{\mathbf{M}}(x,t;\lambda)$ that satisfies the following properties:
		\begin{itemize}
			\item \textbf{Analyticity:} $\hat{\mathbf{M}}(x,t;\lambda)$ is analytic in $\lambda$ for $\lambda\in\mathbb{C}\backslash\Sigma$, where $\Sigma=\mathbb{R}\cup\{\partial D_{ k},\partial\hat{ D}_{ k}\}_{k=1}^N$;
			\item \textbf{Jump:} The matrix $\hat{\mathbf{M}}(x,t;\lambda)$ takes continuous boundary values $\hat{\mathbf{M}}_{\pm}(x,t;\lambda)$ on $\mathbb{R}$ from $\mathbb{C}^{\pm}$, as well as from the left and right on the boundaries $\{\partial{D}_k\}_{k=1}^N$ and $\{\partial\hat{D}_k\}_{k=1}^N$, where the ``$\pm$" subscripts indicate boundary values on $\Sigma$. The boundary values are related by specific jump conditions,
			\begin{equation}\label{tMjump}
				\hat{\mathbf{M}}_{+}(\lambda;t)=\hat{\mathbf{M}}_{-}(x,t;\lambda)\hat{\mathbf{J}}(x,t;\lambda),\quad \lambda\in\Sigma,
			\end{equation}
			where
			\begin{equation}\label{eq:tjump}
				\hat{\mathbf{J}}(x,t;\lambda)=\begin{cases}\mathbf{J}(x,t;\lambda),\quad &\lambda\in\mathbb{R},\\
					\mathbf{P}_k^{-1}(x,t;\lambda),\quad &\lambda\in\partial D_k,\quad k=1,\ldots,N,\\
					\hat{\mathbf{P}}_k(x,t;\lambda),\quad &\lambda\in\partial \hat{D}_{k},\quad k=1,\ldots,N,
				\end{cases}
			\end{equation}
with $\mathbf{J}(x,t;\lambda)$ is defined in Eq.\,\eqref{jump}, $\mathbf{P}_k(x,t;\lambda)$ and $\hat{\mathbf{P}}_k(x,t;\lambda)$ are defined in Eq.\,\eqref{p};
			\item \textbf{Normalizations:} $\hat{\mathbf{M}}(x,t;\lambda)$ has the following asymptotic behavior:
			\begin{equation}\label{1}\hat{\mathbf{M}}(x,t;\lambda)=\mathbf{I}+O(\lambda^{-1}),\quad\lambda\rightarrow\infty.\end{equation}
		\end{itemize}
	\end{problem}
	The solution of the pole-free RH problem \ref{pro2} exists uniquely if and only if the following lemma is satisfied (see Theorem 9.3 in Ref.\cite{Zhou1989}).		
\begin{lemma}{(Vanishing Lemma)}\label{vanish}
	The RH problem \ref{pro2} has only the zero solution by replacing the asymptotic condition \eqref{1} with
	\begin{alignat*}{2}
		\hat{\mathbf{M}}(x,t;\lambda)=O(\lambda^{-1}), \quad\lambda\rightarrow\infty.
	\end{alignat*}
\end{lemma}
\begin{proof}
	Let \begin{equation}
		\mathcal{M}(x,t;\lambda)=\hat{\mathbf{M}}(x,t;\lambda)\left(\begin{array}{cc}\lambda^{-\frac12}&0\\0&\lambda^{\frac12}\end{array}\right)\hat{\mathbf{M}}^\dag(x,t;\bar\lambda),
	\end{equation} where the notation ``\dag'' denotes the  conjugate transpose. The function
	$\mathcal{M}(x,t;\lambda)$ is analytic for $\lambda\in\mathbb{C}\backslash\Sigma$ and is continuous up to the boundary $\Sigma$. To analyze the jump of $\mathcal{M}(x,t;\lambda)$ across $\Sigma$, we start by noting that  \begin{equation}
		\mathcal{M}_{+}(x,t;\lambda)=\hat{\mathbf{M}}_{+}(x,t;\lambda)\left(\begin{array}{cc}\lambda^{-\frac12}&0\\0&\lambda^{\frac12}\end{array}\right)\hat{\mathbf{M}}_{-}^\dag(x,t;\bar\lambda).
	\end{equation}
It is important to observe that if $\lambda\rightarrow\partial D_k$ from the left (``$+$"  side), the complex conjugate $\bar\lambda\rightarrow\partial\hat D_k$
	from the right (``$-$"  side). By applying the jump conditions across $\Sigma$ as given in Eq.\,\eqref{tMjump}, we have,
	\begin{equation}\label{eq:Hp}\begin{aligned}
		\mathcal{M}_{+}(x,t;\lambda)&=\hat{\mathbf{M}}_{+}(x,t;\lambda)\left(\begin{array}{cc}\lambda^{-\frac12}&0\\0&\lambda^{\frac12}\end{array}\right)\hat{\mathbf{M}}_{-}^\dag(x,t;\bar\lambda)\\
		&=\hat{\mathbf{M}}_{-}(x,t;\lambda)\hat{\mathbf{J}}(x,t;\lambda)\left(\begin{array}{cc}\lambda^{-\frac12}&0\\0&\lambda^{\frac12}\end{array}\right)\hat{\mathbf{M}}_{-}^\dag(x,t;\bar{\lambda}).
	\end{aligned}\end{equation}
Similarly, we can express $\mathcal{M}_{-}(x,t;\lambda)$:
\begin{equation}\label{eq:Hp2}\begin{aligned}
		\mathcal{M}_{-}(x,t;\lambda)&=\hat{\mathbf{M}}_{-}(x,t;\lambda)\left(\begin{array}{cc}\lambda^{-\frac12}&0\\0&\lambda^{\frac12}\end{array}\right)\hat{\mathbf{M}}_{+}^\dag(x,t;\bar\lambda)\\
		&=\hat{\mathbf{M}}_{-}(x,t;\lambda)\left(\begin{array}{cc}\lambda^{-\frac12}&0\\0&\lambda^{\frac12}\end{array}\right)\hat{\mathbf{J}}^\dag(x,t;\bar\lambda)\hat{\mathbf{M}}_{-}^\dag(x,t;\bar\lambda).
\end{aligned}\end{equation}
	Next, utilizing the relationship
	\begin{equation}
		\hat{\mathbf{J}}(x,t;\lambda)\left(\begin{array}{cc}\lambda^{-\frac12}&0\\0&\lambda^{\frac12}\end{array}\right)=\left(\begin{array}{cc}\lambda^{-\frac12}&0\\0&\lambda^{\frac12}\end{array}\right)\hat{\mathbf{J}}^\dag(x,t;\bar\lambda),\quad\lambda\in\Sigma,
	\end{equation}
	we find $\mathcal{M}_{+}(x,t;\lambda)=\mathcal{M}_{-}(x,t;\lambda)$, which imply $\mathcal{M}(x,t;\lambda)$ is continuous in the whole complex $\lambda$-plane. It then follows by Morera's Theorem that $\mathcal{M}(x,t;\lambda)$ is an entire function of $\lambda$.
	Furthermore, since $\lim_{\lambda\rightarrow\infty}\mathcal{M}(x,t;\lambda)=\mathbf{0}$,
	Liouville's Theorem allows us to conclude that
	\begin{equation}
		\mathcal{M}(x,t;\lambda)\equiv\mathbf{0},\quad\lambda\in\mathbb{C}.
	\end{equation}
	
In the following, we investigate the value of $\hat{\mathbf{M}}(x,t;\lambda)$. As $\lambda \in \mathbb{R}^+$,	
	\begin{equation}\begin{aligned}
		\hat{\mathbf{J}}(x,t;\lambda)\left(\begin{array}{cc}\lambda^{-\frac12}&0\\0&\lambda^{\frac12}\end{array}\right)
		=\left(\begin{array}{cc}\lambda^{-\frac14}&0\\0&\lambda^{\frac14}\end{array}\right)\widetilde{\mathbf{J}}(x,t;\lambda)\left(\begin{array}{cc}\lambda^{-\frac14}&0\\0&\lambda^{\frac14}\end{array}\right),
	\end{aligned}\end{equation}
where
\begin{equation}
	\widetilde{\mathbf{J}}(x,t;\lambda)=\left(\begin{array}{cc}1&\lambda^{\frac12}\gamma(\lambda)\mathrm{e}^{2\mathrm{i}\Theta(\lambda)}\\\lambda^{\frac12}\overline{\gamma(\bar{\lambda})}\mathrm{e}^{-2\mathrm{i}\Theta(\lambda)}&1+\lambda|\gamma(\lambda)|^2\end{array}\right).
\end{equation}
Subsequently, we find that
\begin{equation}\label{M}\begin{aligned}
		\mathcal{M}_{+}(x,t;\lambda)&=\left(\hat{\mathbf{M}}_{-}(x,t;\lambda)\left(\begin{array}{cc}\lambda^{-\frac14}&0\\0&\lambda^{\frac14}\end{array}\right)\right)\widetilde{\mathbf{J}}(x,t;\lambda)\left(\hat{\mathbf{M}}_{-}(x,t;\lambda)\left(\begin{array}{cc}\lambda^{-\frac14}&0\\0&\lambda^{\frac14}\end{array}\right)\right)^\dag\\
		&=\mathbf{0}.
\end{aligned}\end{equation}	
From this equation, we conclude that since
	$\widetilde{\mathbf{J}}(x,t;\lambda)$ is positive definite for $\lambda\in\mathbb{R}^+$, it follows that  $\hat{\mathbf{M}}_{-}(x,t;\lambda)=\mathbf{0}$ for $\lambda\in\mathbb{R}^+$. Applying the jump condition leads us to conclude that $\hat{\mathbf{M}}_{+}(x,t;\lambda)=\mathbf{0}$ for $\lambda\in\mathbb{R}^+$. Therefore, we can assert that $\hat{\mathbf{M}}(x,t;\lambda)\equiv\mathbf{0}$ for $\lambda\in\mathbb{R}^+$.
Based on the uniqueness theorem for analytic functions, we have $\hat{\mathbf{M}}(x,t;\lambda)=\mathbf 0$ holds for all $\lambda \in\mathbb{C}\backslash\{\{D_{k},\hat{ D}_{k}\}_{k=1}^N\cup\mathbb{R}^-\}$. Specifically, since $\hat{\mathbf{M}}_{+}(x,t;\lambda)=\mathbf {0}$ on $\{\partial D_k\}_{k=1}^N$, and by applying the jump condition on $\{D_{k}\}_{k=1}^N$, we deduce that $\hat{\mathbf{M}}(x,t;\lambda)=\mathbf{0}$ for $\lambda$ within the interior of $\{D_{ k}\}_{k=1}^N$. Thus, $\hat{\mathbf{M}}(x,t;\lambda)=\mathbf{0}$ for $\lambda\in \{D_k\}_{k=1}^N$. Similarly, we establish that $\hat{\mathbf{M}}(x,t;\lambda)=\mathbf{0}$ for $\lambda\in \{\hat{D}_k\}_{k=1}^N\cup\mathbb{R}^-$. In conclusion, we have demonstrated that $\hat{\mathbf{M}}(x,t;\lambda)=\mathbf{0}$ throughout the entire complex plane, as desired.
\end{proof}
	
	\subsection{Reconstruction of potential}
	\begin{theorem}\label{reconstruction}(Reconstruction formula) Let $\hat{\mathbf{M}}(x,t; \lambda)$ be the solution to RH problem $\ref{pro2}$. Then the solution $u(x,t)$ and $v(x,t)$ of MTM \eqref{mtm} can be reconstructed as follows:
		\begin{equation}\label{solution}
			u(x,t)=\overline{\lim_{\lambda\to0}\hat{\mathbf{M}}(x,t;\lambda)_{12}},\qquad v(x,t)=\lim_{\lambda\to0}\hat{\mathbf{M}}(x,t;\lambda)_{21}.
		\end{equation}
	\end{theorem}
	\begin{proof}
		This result can be demonstrated using the dressing method as detailed in Ref.\cite{dressing}.
	\end{proof}
By combining  Theorem \ref{reconstruction} with the definition given in \eqref{hatM}, we arrive at the following statement.
	\begin{corollary}\label{recon}
	Let $\mathbf{M}(x,t; \lambda)$ be the solution to RH problem $\ref{pro}$. Then the solution $u(x,t)$ and $v(x,t)$ of MTM \eqref{mtm} can be reconstructed as follows:
	\begin{equation}\label{solution2}
		u(x,t)=\overline{\lim_{\lambda\to0}\mathbf{M}(x,t;\lambda)_{12}},\qquad v(x,t)=\lim_{\lambda\to0}\mathbf{M}(x,t;\lambda)_{21}.
	\end{equation}
	\end{corollary}
	
	\begin{remark}
		In the present context, the reconstruction formulas for $u(x,t)$ and $v(x,t)$ are elucidated with simplicity and clarity, which contrasts favorably with the formulations provided in  Eqs.\,(4.5) and (4.10)  of Ref.\,\cite{inverse}.
	\end{remark}
	The regularization of the RH problem $\ref {pro}$ can be achieved by subtracting pole contributions and the leading order asymptotics at infinity from the matrix $\mathbf{M}(x,t;\lambda)$. The resulting matrix $\mathscr{M}(x,t;\lambda)$ is given by:
	\begin{equation}\begin{aligned}
			\mathscr{M}(x,t;\lambda)=& \mathbf{M}(x,t;\lambda)-\mathbf{I} \\
			&\begin{aligned}-\sum_{k=1}^{N}\sum_{n_{k}=0}^{m_{k}}\left(\frac{\underset{\lambda_k}{\operatorname*{Res}} (\lambda-\lambda_{k})^{n_{k}}\mathbf{M}(x,t;\lambda)}{(\lambda-\lambda_{k})^{n_{k}+1}}+\frac{\underset{\bar{\lambda}_k}{\operatorname*{Res}} (\lambda-\bar{\lambda}_{k})^{n_{k}}\mathbf{M}(x,t;\lambda)}{(\lambda-\bar{\lambda}_{k})^{n_{k}+1}}\right).\end{aligned}
	\end{aligned}\end{equation}
	This regularized matrix $\mathscr{M}(x,t;\lambda)$ satisfies the following conditions:
	\begin{equation}\begin{aligned}&\mathscr{M}_{+}(x,t;\lambda)-\mathscr{M}_{-}(x,t;\lambda)=\mathbf{M}_{-}(x,t;\lambda)(\mathbf{J}(x,t;\lambda)-\mathbf{I}),&&\lambda\in\mathbb{R}\backslash\{0\},\\&\mathscr{M}(x,t;\lambda)\to\mathbf{0},&&\lambda\to\infty.\end{aligned}\end{equation}
	By applying the Sokhotski–Plemelj formula, we obtain
	\begin{equation}\mathscr{M}(x,t;\lambda)=\frac1{2\pi\mathrm{i}}\int_\mathbb{R}\frac{\mathbf{M}_-(x,t;\xi)(\mathbf{J}(x,t;\xi)-\mathbf{I})}{\xi-\lambda}\mathrm{d}\xi.\end{equation}
	The RH problem $\ref{pro}$ can thus be solved by the system of algebraic-integral equations:
	\begin{equation}\label{m11}\begin{aligned}
			\mathbf{M}_1(x,t;\lambda) =&\begin{pmatrix}1\\0\end{pmatrix}-\sum_{k=1}^N\sum_{n_k=0}^{m_k}\frac{[\lambda\overline{f(\bar{\lambda})}\mathrm{e}^{-2\mathrm{i}\Theta(x,t;\lambda)}\mathbf{M}_2(x,t;\lambda)]^{(m_k-n_k)}|_{\lambda=\bar{\lambda}_k}}{(\lambda-\bar{\lambda}_k)^{n_k+1}(m_k-n_k)!} \\
			&+\frac1{2\pi\mathrm{i}}\int_\mathbb{R}\frac{\mathbf{M}_-(x,t;\xi)(\mathbf{J}(x,t;\xi)-\mathbf{I})_1}{\xi-\lambda}\mathrm{d}\xi, \end{aligned}\end{equation}
	\begin{equation}
		\label{22}\begin{aligned}
			\mathbf{M}_2(x,t;\lambda) =&\begin{pmatrix}0\\1\end{pmatrix}+\sum_{k=1}^N\sum_{n_k=0}^{m_k}\frac{[f(\lambda)\mathrm{e}^{2\mathrm{i}\Theta(x,t;\lambda)}\mathbf{M}_1(x,t;\lambda)]^{(m_k-n_k)}|_{\lambda=\lambda_k}}{(\lambda-\lambda_k)^{n_k+1}(m_k-n_k)!} \\
			&+\frac1{2\pi\mathrm{i}}\int_{\mathbb{R}}\frac{\mathbf{M}_-(x,t;\xi)(\mathbf{J}(x,t;\xi)-\mathbf{I})_2}{\xi-\lambda}\mathrm{d}\xi .
	\end{aligned}\end{equation}
	\begin{corollary}
		Based on Eqs.\,\eqref{m11} and \eqref{22}, $u(x,t)$ and $v(x,t)$ in \eqref{solution} can be simplified as
		\begin{equation}\begin{aligned}
				u(x,t)=&\overline{\sum_{k=1}^{N}\sum_{n_k=0}^{m_{k}}\frac{[f(\lambda)\mathrm{e}^{2\mathrm{i}\Theta(x,t;\lambda)}\mathbf{M}_{11}(x,t;\lambda)]^{(m_{k}-n_{k})}|_{\lambda=\lambda_{k}}}{(-\lambda_{k})^{n_{k}+1}(m_{k}-n_{k})!}}\\
				&+\overline{\frac1{2\pi\mathrm{i}}\int_\mathbb{R}\frac{\mathbf{M}_-(x,t;\xi)(\mathbf{J}(x,t;\xi)-\mathbf{I})_{12}}{\xi}\mathrm{d}\xi},\end{aligned}\end{equation}and
		\begin{equation}\begin{aligned}v(x,t)=&-\sum_{k=1}^{N}\sum_{n_k=0}^{m_{k}}\frac{[\lambda\overline{f(\bar{\lambda})}\mathrm{e}^{-2\mathrm{i}\Theta(x,t;\lambda)}\mathbf{M}_{22}(x,t;\lambda)]^{(m_{k}-n_{k})}|_{\lambda=\overline{\lambda}_{k}}}{(-\overline{\lambda}_{k})^{n_{k}+1}(m_{k}-n_{k})!}\\
				&+\frac1{2\pi\mathrm{i}}\int_\mathbb{R}\frac{\mathbf{M}_-(x,t;\xi)(\mathbf{J}(x,t;\xi)-\mathbf{I})_{21}}{\xi}\mathrm{d}\xi.
		\end{aligned}\end{equation}
	\end{corollary}
	\subsection{Reflectionless potential}
	We now explicitly reconstruct the potentials $u(x,t)$ and $v(x,t)$ in the reflectionless scenario, i.e., $\gamma(\lambda)=0$. Here, there is no jump across the contour $\mathbb{R}$, leading to a reduction of the inverse problem to the algebraic system
	\begin{alignat}{2}
		&\mathbf{M}_1(x,t;\lambda)=\begin{pmatrix}1\\0\end{pmatrix}-\sum_{k=1}^N\frac1{m_k!}\partial_\mu^{m_k}\left.\left[\frac{\mu\overline{f(\bar{\mu})}\mathrm{e}^{-2\mathrm{i}\Theta(x,t;\mu)}\mathbf{M}_2(x,t;\mu)}{\lambda-\mu}\right]\right|_{\mu=\bar{\lambda}_k},\label{m1}\\&\mathbf{M}_2(x,t;\lambda)=\begin{pmatrix}0\\1\end{pmatrix}+\sum_{l=1}^N\frac1{m_l!}\partial_\nu^{m_l}\left.\left[\frac{f(\nu)\mathrm{e}^{2\mathrm{i}\Theta(x,t;\nu)}\mathbf{M}_1(x,t;\nu)}{\lambda-\nu}\right]\right|_{\nu=\lambda_l}.\label{m2}
	\end{alignat}
	Let
	\begin{equation}\begin{aligned}
			&F_1(x,t;\lambda)=f(\lambda)\mathrm{e}^{2\mathrm{i}\Theta(x,t;\lambda)}\mathbf{M}_{11}(x,t;\lambda),\\ &F_2(x,t;\lambda)=\lambda\overline{f(\bar{\lambda})}\mathrm{e}^{-2\mathrm{i}\Theta(x,t;\lambda)}\mathbf{M}_{22}(x,t;\lambda).
	\end{aligned}\end{equation}By substituting $\mathbf{M}_2(x,t;\mu)$ into Eq.\,\eqref{m1}, we obtain the following expression which holds for fixed $(x,t)\in\mathbb{R}\times\mathbb{R}^+$ and $\lambda\in\mathbb{C}\backslash\{0,\lambda_1,\dots,\lambda_N\}$,
	\begin{equation}\begin{aligned}
			G_1(x,t;\lambda)&=F_1(x,t;\lambda)-f(\lambda)\mathrm{e}^{2\mathrm{i}\Theta(x,t;\lambda)} \\
			&\left.+\sum_{k=1}^N\sum_{l=1}^N\frac{1}{m_k!m_l!}\partial_\mu^{m_k}\partial_\nu^{m_l}\left[\frac{\mu f(\lambda)\overline{f(\bar{\mu})}\mathrm{e}^{2\mathrm{i}(\Theta(x,t;\lambda)-\Theta(x,t;\mu))} F_1(x,t;\nu)}{(\lambda-\mu)(\mu-\nu)}\right]\right|_{\substack{\mu=\bar{\lambda}_k\\\nu=\lambda_l}}\\
			&=0.
	\end{aligned}\end{equation}
	Similarly to the above discussion, substituting $\mathbf{M}_1(x,t;\nu)$ into Eq.\,\eqref{m2}, we obtain the following expression which holds for fixed $(x,t)\in\mathbb{R}\times\mathbb{R}^+$ and $\lambda\in\mathbb{C}\backslash\{0,\bar{\lambda}_1,\dots,\bar{\lambda}_N\}$,
	\begin{equation}\begin{aligned}
			G_2(x,t;\lambda)&=F_2(x,t;\lambda)-\lambda\overline{f(\bar{\lambda})}\mathrm{e}^{-2\mathrm{i}\Theta(x,t;\lambda)} \\
			&\left.+\sum_{k=1}^N\sum_{l=1}^N\frac{1}{m_k!m_l!}\partial_\nu^{m_l}\partial_\mu^{m_k}\left[\frac{\lambda \overline{f(\bar{\lambda})}f(\nu)\mathrm{e}^{2\mathrm{i}(\Theta(x,t;\nu)-\Theta(x,t;\lambda))} F_2(x,t;\mu)}{(\lambda-\nu)(\nu-\mu)}\right]\right|_{\substack{\mu=\bar{\lambda}_k\\\nu=\lambda_l}}\\
			&=0.
	\end{aligned}\end{equation}
	\begin{theorem}
		In the reflectionless case, the solution of the MTM \eqref{mtm} can be expressed by
		\begin{equation}\begin{aligned}
				&u(x,t)=\overline{\sum_{k=1}^N\sum_{n_k=0}^{m_k}\frac{F_1^{(m_k-n_k)}(\lambda_k)}{(-\lambda_k)^{n_k+1}(m_k-n_k)!}},\\
				&v(x,t)=-\sum_{k=1}^N\sum_{n_k=0}^{m_k}\frac{F_2^{(m_k-n_k)}(\bar{\lambda}_k)}{(-\bar{\lambda}_k)^{n_k+1}(m_k-n_k)!},
		\end{aligned}\end{equation}
		where $\{F_1^{(j_k)}(x,t;\lambda_k),F_2^{(j_k)}(x,t;\bar{\lambda}_k)\}_{k=1,\dots,N}^{j_k=0,\dots,m_k}$ is the solution of the following algebraic systems, respectively,
		\begin{equation}\begin{cases}\frac{G_1^{(j_1)}(x,t;\lambda_1)}{j_1!}=0,&\frac{G_2^{(j_1)}(x,t;\bar{\lambda}_1)}{j_1!}=0,\quad j_1=0,\ldots,m_1,\\\qquad\quad\qquad\qquad\vdots\\ \frac{G_1^{(j_N)}(x,t;\lambda_N)}{j_N!}=0,&\frac{G_2^{(j_N)}(x,t;\bar{\lambda}_N)}{j_N!}=0,\quad j_N=0,\ldots,m_N.\end{cases}
		\end{equation}
	\end{theorem}
	\subsection{Some explicit solutions}\label{numer}
	In the following, we explore numerical simulations of the $N$-multipole solutions for the MTM \eqref{mtm} with different parameter values.
	\begin{enumerate}[label= , left=0pt]
		\item \textbf{Case (i)} $N$=1, three 1-multipole solutions, as shown in Figure \ref{fig3}.
		\begin{figure}[!htbp]
			\centering
			\includegraphics[width=0.94\textwidth,height=0.24\textheight]{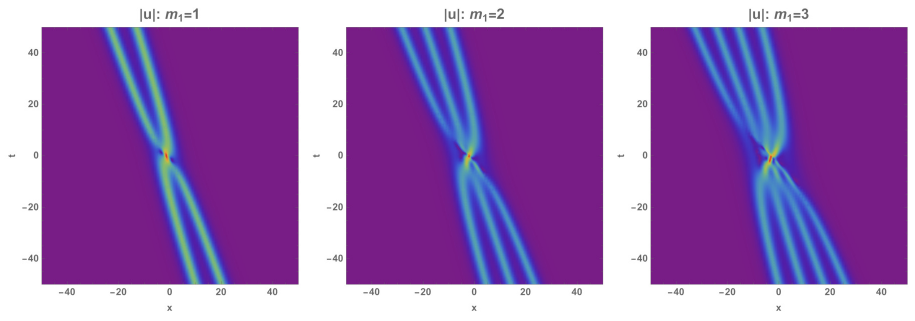}
			\includegraphics[width=0.94\textwidth,height=0.24\textheight]{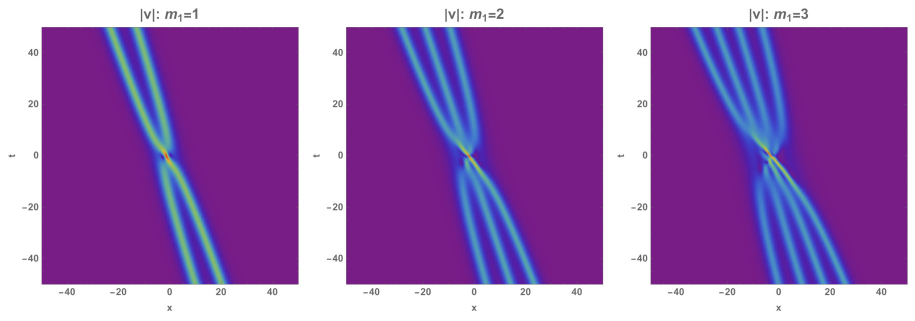}
\caption{ Density structures of $|u(x,t)|$ and $|v(x,t)|$ (from top to bottom) for $\lambda_1 = \mathrm{i} + 1, f(\lambda)=1$ and $m_1 = 1, 2, 3$ (from left to right).}\label{fig3}
		\end{figure}
		\item \textbf{Case (ii)} $N$=2, three 2-multipole solutions, as shown in Figure \ref{fig4}.
		\begin{figure}[!htbp]
			\centering
			\includegraphics[width=0.94\textwidth,height=0.24\textheight]{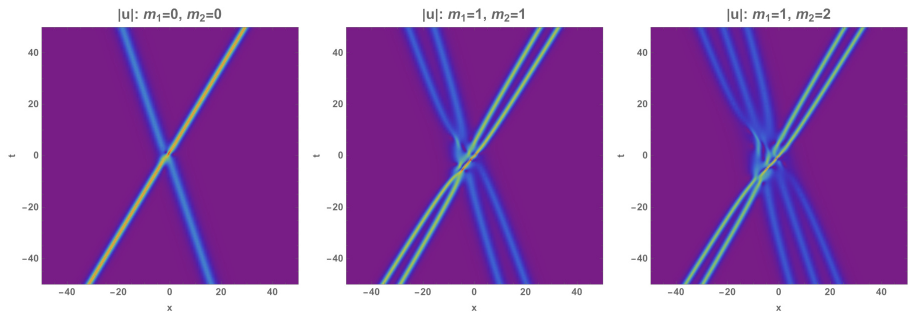}
			\includegraphics[width=0.94\textwidth,height=0.24\textheight]{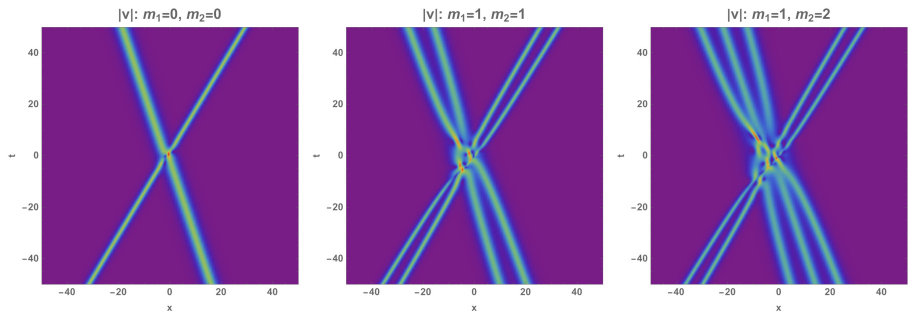}
\caption{Density structures of $|u(x,t)|$ and $|v(x,t)|$ (from top to bottom) for $\lambda_1 = \frac{\mathrm{i}}{2},\lambda_2=\mathrm{i}+1, f(\lambda)=\lambda-\mathrm{i}$ and $m_1 = 0, m_2=0$ (left), $m_1 = 1, m_2=1$ (middle), $m_1 = 1, m_2=2$ (right).}\label{fig4}
		\end{figure}
		
		\item \textbf{Case (iii)} $N$=3, three 3-multipole solutions, as shown in Figure \ref{fig5}.
		\begin{figure}[!htbp]
			\centering
			\includegraphics[width=0.94\textwidth,height=0.24\textheight]{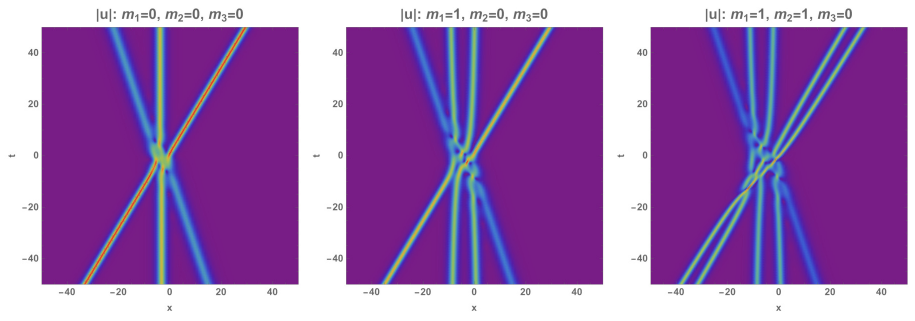}
			\includegraphics[width=0.94\textwidth,height=0.24\textheight]{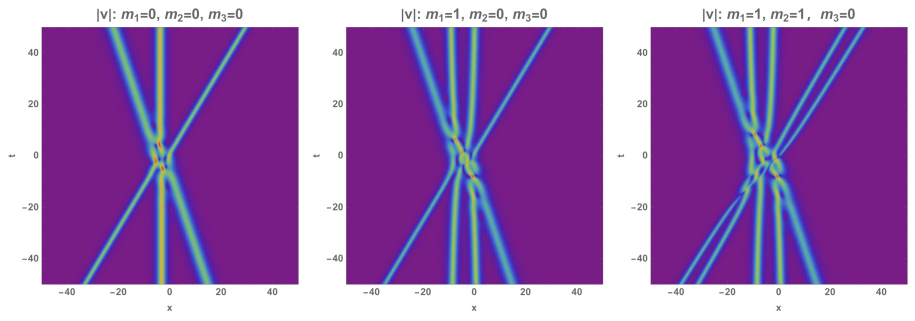}
\caption{ Density structures of $|u(x,t)|$ and $|v(x,t)|$ (from top to bottom) for $\lambda_1 =\mathrm{i}, \lambda_2=\frac{\mathrm{i}}{2}, \lambda_3=\mathrm{i}+1, f(\lambda)=\mathrm{e}^{\mathrm{i}\lambda}$ and $m_1 = 0, m_2=0, m_3 = 0$ (left), $m_1 = 1, m_2=0, m_3 = 0$ (middle), $m_1 = 1, m_2=1, m_3 = 0$ (right).}\label{fig5}
		\end{figure}
		
	\end{enumerate}
	\section*{Acknowledgment}
	This work was supported by National Natural Science Foundation of China (Grant Nos. 12371253 and 12171439).	
	
	\section*{Data availability}
	All data generated or analyzed during this study are including in this published article.	
	
	\section*{Declarations}	
	
	\section*{Conflict of Interest}		
	The authors declare that they have no conflict of interest.

\end{document}